\documentclass[USenglish]{lipics-v2021}
\nolinenumbers
\hideLIPIcs

\usepackage{blindtext}

\usepackage{xspace}
\usepackage{graphicx}
\graphicspath{ {./figures/} }
\usepackage{pgfplots}
\pgfplotsset{width=10cm,compat=1.9}
\usepackage[export]{adjustbox}
\usepackage{listings}
\usepackage{mathpartir}
\usepackage[linesnumbered,ruled,vlined]{algorithm2e}
\usepackage{tikz}
\usepackage{tikz-cd}
\usetikzlibrary{decorations.pathmorphing,positioning,arrows.meta,positioning,calc, fit}
\usepackage{xcolor}
\usepackage{siunitx}
\usepackage{upgreek}
\usepackage{mathtools}

\definecolor{linkBlue}{HTML}{0044CC}
\hypersetup{colorlinks,linkcolor=.,citecolor=.,filecolor=.,urlcolor=linkBlue}

\usepackage[capitalise]{cleveref}

\definecolor{lightlightgray}{HTML}{dedede}

\title{Remote Concolic Multiverse Debugging}
\subtitle{Extended Version with Additional Appendices}

\author{Maarten Steevens}{Ghent University, Belgium}
{maarten.steevens@ugent.be}
{0009-0002-6339-0467}{Funded by the Research Foundation Flanders, grant number 1SA6C26N}

\author{Tom Lauwaerts}{Vrije Universiteit Brussel, Belgium}
{tom.lauwaerts@vub.be}
{0000-0003-1262-8893}{Funded by a project (grant number FWOOPR2020008201)
from the Research Foundation Flanders during the period in which this research was conducted.}

\author{Christophe Scholliers}{Ghent University, Belgium}
{christophe.scholliers@ugent.be}
{0000-0002-2837-4763}{}

\authorrunning{M. Steevens, T. Lauwaerts, and C. Scholliers}

\Copyright{Maarten Steevens, Tom Lauwaerts, and Christophe Scholliers}

\keywords{Multiverse Debugging, Embedded devices, WebAssembly}
    
\definecolor{mRed}{rgb}{1,0.4,0.4}
\definecolor{mBlue}{rgb}{0.4,0.4,1}
\definecolor{mGreenGraph}{RGB}{102,255,110}
\definecolor{mGreen}{rgb}{0,0.6,0}
\definecolor{mGray}{rgb}{0.5,0.5,0.5}
\definecolor{mDarkGray}{rgb}{0.2,0.2,0.4}
\definecolor{mPurple}{rgb}{0.58,0,0.82}
\definecolor{backgroundColour}{rgb}{1,1,1}

\lstdefinestyle{CStyle}%
{     basicstyle=\footnotesize\ttfamily\linespread{0.7}%
, captionpos=b%
, identifierstyle=%
, backgroundcolor=\color{backgroundColour}
, commentstyle=\color{mGray}
, keywordstyle=\bfseries\color{black}
, numberstyle=\tiny\ttfamily\color{mGray}
, stringstyle=\color{mRed}
, keywordstyle=\color{mBlue}\bfseries
, keywordstyle=[2]\color{mRed}
, keywordstyle=[3]\color{mBlue}
, keywordstyle=[4]\color{mRed}
, keywordstyle=[5]\color{mBlue}
, columns=spaceflexible%
, keepspaces=true%
, showspaces=false%
, showtabs=false%
, showstringspaces=false%
, numbers=left%
, numbersep=5pt%
}

\lstdefinelanguage{AssemblyScript}{
sensitive,
morecomment=[l]{//},
morekeywords={import, export, let, const, while, class, function, as, from, enum},
morekeywords=[2]{true, void, u32, i32, boolean, Pin, Color, string, Options},
morekeywords=[3]{@external},
morestring=[b]{"},
morestring=[b]` 
}

\EventEditors{}
\EventNoEds{0}
\EventLongTitle{}
\EventShortTitle{}
\EventAcronym{}
\EventYear{}
\EventDate{}
\EventLocation{}
\EventLogo{}
\SeriesVolume{}
\ArticleNo{}

\begin{document}

\ccsdesc[300]{Computer systems organization~Embedded software}
\ccsdesc[500]{Software and its engineering~Semantics}
\ccsdesc[500]{Software and its engineering~Software testing and debugging}
\ccsdesc[300]{Software and its engineering~Integrated and visual development environments}

\maketitle

\begin{abstract}
Debugging nondeterministic programs is inherently difficult, particularly in microcontroller environments where execution paths can diverge unpredictably due to external sensor inputs. 
Traditional debugging techniques often fail to capture or reproduce this nondeterministic behavior effectively.
Multiverse debugging has emerged as a compelling technique to debug nondeterministic programs, allowing developers to systematically explore all possible execution paths.
Unfortunately, current multiverse debuggers are snapshot-based and most operate over a model of the program, which limits their use for debugging resource-constrained microcontrollers.
Additionally, current multiverse debuggers, even ones specifically designed for microcontrollers suffer from state explosion making the state space overwhelming during debugging.

To address these challenges, we introduce a trace-based multiverse debugger with a novel state-space reduction technique based on concolic execution.
Our approach interleaves concolic analysis with live debugging to identify input values that define unique program paths. 
This hybrid technique efficiently prunes redundant paths from the state space while ensuring full code coverage.
Unlike MIO, a recently published multiverse debugger for microcontrollers that focuses on IO consistency, our approach directly targets state explosion by leveraging concolic execution and uses a trace-based approach, significantly reducing the memory and communication overhead.

We implemented a prototype using the WARDuino WebAssembly virtual machine on an STM32 microcontroller, demonstrating the feasibility and efficiency of our approach in real-world scenarios. 
Our results highlight substantial reductions in the state space compared to traditional multiverse debugging. 
This makes multiverse debugging more accessible and efficient for developers working with complex, nondeterministic programs running on microcontrollers.
\end{abstract}

\section{Introduction}
Developers naturally debug programs by iteratively verifying or falsifying a hypothesis by examining the program's execution~\cite{zeller05, perscheid17}.
Nondeterministic programs complicate this method significantly~\cite{mcdowell89, gurdeep22}, particularly in microcontroller applications where bugs can manifest unpredictably due to external interactions.  
Traditional debuggers struggle to reliably and efficiently reproduce the circumstances in which bugs manifest.
Most existing debuggers do not provide developers with tools to manage and explore the different execution paths of the program being debugged.
One notable exception is a technique called multiverse debugging~\cite{torres19} which has emerged as a promising technique to enable programmers to browse through all possible execution paths. 
Most existing multiverse debuggers unfortunately operate over a model of the program instead of the concrete execution~\cite{torres19,pasquier22,pasquier23,pasquier23a}, and are fundamentally \emph{offline} techniques except for the recent MIO~\cite{mio} debugger. 
Additionally, browsing through all execution paths results in a state-explosion making current debuggers impractical.
To create an efficient debugger capable of dealing with state-explosion in a constraint microcontroller setting, we were confronted with the three main challenges.

First, for programs driven by nondeterministic sensor inputs, the number of potential execution paths increases exponentially. 
Each possible sensor value introduces a new branch in the execution tree that the programmer must then meticulously track and evaluate to determine what input combinations are interesting. 
This results in an overwhelming state space, making it computationally and mentally intractable to exhaustively explore all possibilities~\cite{liEmpiricalStudyConcurrency2023a}. 
Even systems with limited nondeterministic input values, for example an application with a simple temperature and light sensor\footnote{With a moderate precision 12 bit ADC each sensor has 4096 possible values, meaning that there are more than 16 million value pairs to consider.}, are already susceptible to this issue. 

Second, current multiverse debuggers are very resource-intensive due to their snapshot-based approach.
Each of these snapshots represents a complete state of the program, including all memory contents. 
The sheer volume of data required to represent these snapshots quickly becomes prohibitive, particularly when considering the limited resources available on microcontrollers.

Third, current offline techniques do not align with developers preferences for debugging on the concrete hardware \cite{makhshari21} to avoid inaccuracies from simulators or approximations~\cite{roska90,khan11}.

In this work, we present a solution to these challenges: a \emph{remote concolic} multiverse debugger that focusses on pruning the state-space in multiverse debuggers instead of dealing with IO consistency like previous approaches such as MIO~\cite{mio}.

First, the central innovation of this work is the integration of a static analysis technique, concolic execution, into the multiverse debugging process to mitigate state explosion. 
Traditional multiverse debugging~\cite{torres19,mio} relies on exhaustive exploration of all sensor inputs. 
In contrast, our approach performs online concolic execution on a more computationally powerful remote host to analyze the control flow.
The symbolic analysis identifies critical input values necessary to cover all program paths.
These inputs can then be used by the debugger to explore the concrete execution.
The key advantage stems from this hybrid methodology: the static analysis intelligently prunes the possible inputs, and provides example values for each possible path to the debugger to explore.
This combination significantly reduces the developer's cognitive load as the debugger can \emph{visually present} only the essential execution paths, thereby simplifying the analysis of complex nondeterministic behavior.

Second, we address the high resource consumption of conventional multiverse debuggers. 
Existing multiverse debuggers capture the entire program state for each explored state, which requires excessive memory and communication. 
Our approach replaces this with a highly efficient, trace-based approach. 
Instead of storing complete memory snapshots, our debugger records only the sequence of nondeterministic events (e.g., sensor inputs) along the execution paths. 
This minimal trace provides all the information necessary to deterministically replay any path from an initial state, significantly reducing communication and memory overhead. 

Finally, just like MIO~\cite{mio} we have chosen to build an online debugger that operates directly on the embedded device, catering for the embedded software developer preferences.

\subsection{Contributions}
We introduce a trace-based remote concolic multiverse debugger that employs concolic execution to intelligently guide the exploration of a program's nondeterministic behaviors, directly addressing the critical problem of state explosion.
The realization of this approach led to the following main contributions:

\begin{itemize}
	\item The first practical implementation\footnote{Our prototype is provided as an open-source project built upon the user interface of MIO, and can be found here: \url{https://github.com/TOPLLab/MIO/tree/rcmd}} of a novel hybrid debugging architecture combining concolic execution with multiverse debugging targeting microcontrollers.
	
	\item A novel state-space reduction technique that reduces the state-space of nondeterministic programs in multiverse debuggers.

	\item An alternative implementation strategy for multiverse debuggers that uses a trace-based approach instead of a snapshot-based approach, significantly reducing debugger overhead.

	\item A formal model that defines the interaction between online concolic analysis and trace-based multiverse debugging, grounding our hybrid architecture in a theoretical foundation.
	
	\item A proof showing the correctness of our approach.
	
	\item A quantitive and qualitative evaluation over a set of Arduino programs, demonstrating a substantial reduction of the state space compared to existing multiverse techniques.
\end{itemize}

\section{Hands on Remote Concolic Multiverse Debugging}\label{sec:perspective}

\begin{figure*}
	\centering
	\includegraphics[width=\textwidth]{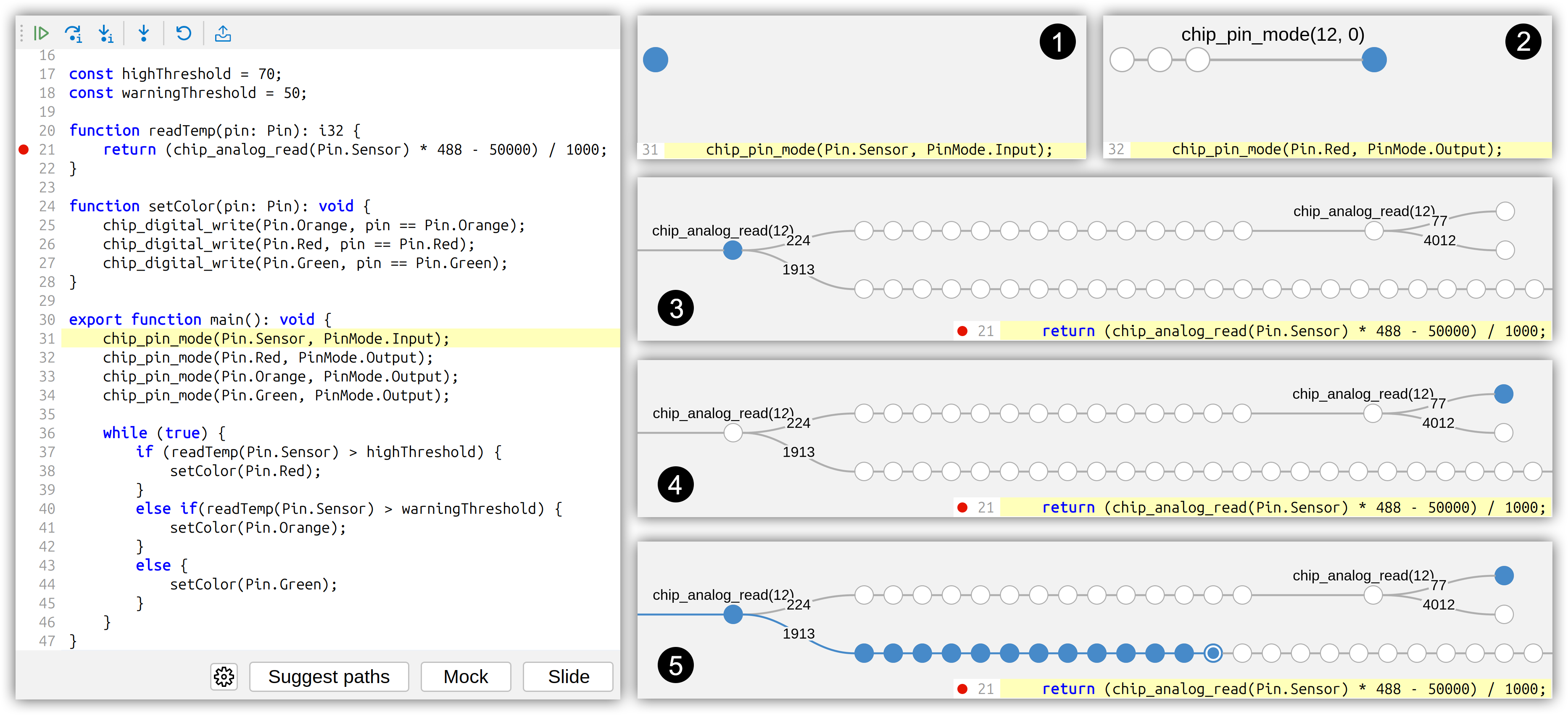}
	\caption{
	    Partial screenshots of the proposed debugger.
      \emph{Left:} the source view. \emph{Right:} the multiverse tree view as it grows and updates during the debugging session.}
        \label{fig:example02}
\end{figure*}

In this section we provide an overview of how our hybrid approach to debugging works in practice, by showing how a small program can be debugged using our prototype.

\subsection{A Temperature Dashboard Application}
In \Cref{fig:example02} we show our prototype being used to debug a temperature dashboard application, shown in the left pane of the debugger.
See \cref{app:prototype} for a full screenshot of our debugger.

This program reads a temperature sensor to control 3 LEDs, leading to three logical paths in the execution.
If the temperature is higher than 70 degrees, a red LED will turn on.
If the temperature is between 50 and 70 degrees an orange warning LED turns on.
Below 50 the LED is set to green.
In this relatively simple example, the temperature is measured using an analog sensor that varies its resistance based on temperature. 
The resulting analog input signal is then converted into degrees Celsius.

This simple program contains a subtle bug that is hard to discover with traditional debuggers. 
Instead of reading the temperature sensor once and then using this value, the program reads the temperature multiple times (line 37 and 40). 
This can lead to situations where the lights are incorrectly switched on or off given the last sensor value. 
In addition to these glitches, reading out the sensor multiple times per iteration slows down the program and consumes more battery power. 
Using our debugger, these issues can be identified.

\subsection{Remote Concolic Multiverse Debugging Session}
\Cref{fig:example02} illustrates the different stages of a remote concolic multiverse debugging session over our example program. 
After loading the debugger and the program, the first line of the source code is highlighted as shown on the left.
The multiverse tree at this point is only a single root node, the entry point of the program (\emph{pane 1}, \cref{fig:example02}).
Stepping forward extends the current path as shown in \emph{pane 2} (\cref{fig:example02}).
Every edge in this path is a single WebAssembly instruction, and edges for primitive calls are labeled with the name and provided arguments.

Imagine the programmer places a breakpoint just before reading the temperature sensor (line 21) and lets the program run until this point.
After the breakpoint is hit, when using a traditional remote debugger, the programmer would proceed through the program code with an arbitrary temperature value (i.e. the current temperature).
Alternatively the developer might manually heat or cool down\footnote{Such as in this GitHub issue where a bug is reproduced using ice packs \url{https://github.com/tobyweston/temperature-machine/issues/13}.} the sensor to guide the program execution.

In contrast, our concolic multiverse debugger helps guide a programmer's exploration using its concolic execution.
The developer can simply ask the debugger for interesting future execution paths by pressing the `Suggest paths' button (\emph{source view}, \cref{fig:example02}).
The debugger then uses concolic execution to identify all future execution paths within a user configurable bound (number of instructions and/or number of nondeterministic operations).
The possible paths found by the analysis are added to the multiverse tree view (\emph{pane 3}, \cref{fig:example02}), showing each of the possible executions as a new branch.
Compared with remote debugging, this feature frees developers from having to manually find the necessary sensor values for specific control flow paths, either through trial and error or by manually analysing the program which quickly becomes too complicated as the program's complexity increases.

Each choice point in our visualisation corresponds with a call to a nondeterministic input primitive, we label the edges starting from a choice point with the corresponding input values.
This gives developers a clear overview of how the input values influence the program's behavior. 
For instance, in \emph{pane 3} (\cref{fig:example02}), the concolic analysis found two possible branches after executing the \textit{chip\_analog\_read(12)} primitive, with input values 224 and 1913.

Subsequently, the developer can manually explore the effects of the nondeterministic input, by choosing one of the branches added by the concolic analysis.
Developers do not have to change the temperature manually to change the input value to correspond to their desired branch.
Instead, the debugger provides a mocking functionality through the `mock' button, which launches a new pop-up window.
In the pop-up, developers can specify the return value for a specific input primitive and its specific arguments.
In our example, the developer changes the return value for reading pin 12, to 224, steps forward to the next analog read instruction, mocks it to read value 77,  reaching the state shown in \emph{pane 4} (\cref{fig:example02}).
During execution the debugger uses the mocked values instead of the actual sensor value.

While mocking is an essential feature for exploration, it takes quite some manual effort to provide all the right sensor values when exploring the multiverse graph.
To make it easer to navigate the graph, the developer can simply click on a node and press the `Slide' button, to slide to that particular execution state, thereby entering the desired branch.
When sliding to an entirely new branch with different parents, the debugger restarts the program and mocks the right values at each choice point automatically.
\emph{Pane 5} (\cref{fig:example02}) shows how the debugger highlights the path that will be taken from the root node to the selected node (highlighted with a white circle), before the \emph{Slide} button is pressed.
Note, this operation does not modify the tree, even when restarting the program. All previous branches remain visible.

As the outlined debug session shows, the concolic multiverse debugger improves on the common remote debugging experience by making it easier to explore multiple execution paths and gives the developer a clear visual overview of choice points in the program.
Exploring different execution paths is made easier by the \textit{sliding} feature, allowing developers to go from one branch to another, and by the \textit{mocking} feature, which allows the debugger to easily mock any input action.
However, the advantages go beyond this.

The concolic analysis coupled with the tree visualisation can already point out the cause of a bug at a glance.
For instance, when the branches provided by the analysis differ from what is expected, this usually points to a mistake in the program's logic.
This is certainly the case in our example, where one expects a single choice point with three options for the current temperature.
However, the trace shows two choice points with two branches each.
This is a consequence of reading the analog sensor value \emph{twice} in the if statement (lines 37 and 40) instead of reading it once, and storing it in a variable.
As a result the temperature value can differ between the two if statements, which can potentially cause strange behavior.

\section{Contrast with Traditional Debugging}
The outlined debugging scenario with the remote concolic multiverse debugger, contrasts sharply with both the way developers traditionally debug nondeterministic programs on microcontrollers, and with debugging sessions in current multiverse debuggers.

\subparagraph*{Traditional Debugging.}
Our example application illustrates several difficulties in debugging nondeterministic programs with traditional \emph{linear} debuggers, where only one path of the program can be explored at a time.
These debuggers have three major drawbacks.

First, when searching for a bug developers often restart and execute the program multiple times. 
Unfortunately, with a traditional remote debugger the programmer will possibly explore vastly different execution paths each time.
With remote debuggers it is the developers responsibility to keep track of all the differences and similarities between these executions.
The multiverse tree visualisation in our debugger reduces the cognitive load for developers.

Second, it is often not clear whether all interesting execution paths have been considered during debugging. 
In the example, it is hard to predict which sensor values will turn on specific LEDs (due to the non-trivial conversion function, line 20--22).
In more complex programs this can only get more challenging, as more complex operations are performed.
In contrast, our on-demand concolic analysis provides the minimal set of branches the developer needs to consider, along with concrete example values.

Third, even when we know what inputs are needed to explore our desired execution path, it is difficult to manipulate the environment during debugging, to achieve those exact inputs.
Manipulating other sensors, or manipulating programs with more than one sensor, poses additional challenges.
Luckily in our debugger, the programmer can \emph{slide} to the desired location automatically, while the debugger mocks the needed values along the way.

\subparagraph*{Multiverse Debuggers.}
Most multiverse debuggers are not designed for debugging live programs. 
Instead, they operate on the semantics of the underlying programming language~\cite{torres19,pasquier22,pasquier23,pasquier23a}.
Recently, MIO~\cite{mio}, a debugger focussing on IO consistency, became the first multiverse debugger for concrete executions.
MIO finally makes it possible to debug live programs, but like previous approaches it suffers from state explosion. Even for simple programs, the number of possible execution paths can become unmanageable.

For instance, in a program with two sensor readings, a multiverse debugger might present the user 4096 options at each choice point. 
This results in a total of $4096^2$, possible execution paths. 
The programmer is then tasked with manually identifying specific paths from this vast search space to debug the program. 
In contrast to MIO this work focusses on the state-explosion problem and leverages the power of concolic execution to identify sensor values needed to explore each of the execution paths.
In our example program, this significantly reduces the cognitive burden for the programmer since they no longer need to reason about which raw sensor values, that are later converted to Celsius result in a specific branch.

Additionally, existing multiverse debuggers such as MIO use a snapshot based approach where a snapshot is taken after each IO operation. This enables reversible IO in multiverse debugging. Since this work focusses on state-space reduction and not IO, we took a different approach that trades this reversibility for performance by using a trace-based approach.

\section{Remote Concolic Multiverse Debugging}\label{sec:trace}
\label{sec:remote-concolic-multiverse-debugging}
In this section, we describe the operation of our trace-based multiverse debugger through a small-step semantics defined over stack-based language, specifically WebAssembly \cite{haas17}, as this formalisation is representative for a wide variety of virtual machines.
The formalization abstracts away from the details of the underlying WebAssembly semantics as much as possible.
Our novel system is at heart a remote debugger, where a client and server component exchange messages as described by the small-step semantics.
The client component keeps track of the multiverse tree and performs the concolic analysis, while the server instruments and maintains the concrete runtime to perform the live debugging operations.
The semantics are defined over several configurations, which we will discuss first.

\subsection{Debugger Configuration}

\begin{figure}
  \begin{minipage}[b]{0.50\textwidth}
        \centering
        	$$
	      \begin{array}{ l l c l }
          \multicolumn{3}{l}{\textsf{\textbf{WebAssembly Configuration}}} & \\
	        \emph{(Locals)}           & \rho                       & \Coloneqq & i32 \rightarrow v \\
	      	\emph{(Globals)}          & \updelta                   & \Coloneqq & i32 \rightarrow v \\
	      	\emph{(Stack)}            & st                         & \Coloneqq & i32 \rightarrow v \\
	      	\emph{(Memory)}           & \mu                        & \Coloneqq & i32 \rightarrow v \\
	      	\emph{(Instructions)}     & e                          & \Coloneqq & i32.const \; | \; ...\\
	        \emph{(Program state)}    & K & \Coloneqq & \{ \rho,\updelta,st,\mu,e^* \} \\
	      	\emph{(Primitive tables)} & P^{Out | In}  & \Coloneqq & p^* \\
          \emph{(Input Primitive)}  & P^{In}(j) & \Coloneqq & f : v^* \rightarrow v \\ 
          \emph{(Output Primitive)} & P^{Out}(j) & \Coloneqq & f : v^* \rightarrow \epsilon
	      \end{array}
	      $$
	      \caption{WebAssembly configuration of the remote debugger.}
	      \label{fig:configuration-remote-debugger}
    \end{minipage}
\hfill
\begin{minipage}[b]{0.4\textwidth}
      \centering
	    \begin{mathpar}
        \inferrule[(\textsc{Input-Prim})]
       	    { 
                P^{In}(j) = p \\
                v \in \lfloor p(v_a^*) \rfloor \\
            }
            {
              \{ \rho,\updelta,v^*_a : v^*,\mu, (\textbf{call} \; j) : e^* \} \\
              \hookrightarrow_{i}
              \{ \rho,\updelta, v : v^*,\mu, e^* \} 
            }

            \inferrule[(\textsc{Output-Prim})]
       	    { 
                P^{Out}(j) = p \\
                \lfloor p(v_a^*) \rfloor \\
            }
            {
              \{ \rho,\updelta,v^*_a : v^*,\mu, (\textbf{call} \; j) : e^* \} \\
              \hookrightarrow_{i}
              \{ \rho,\updelta, v^*,\mu, e^* \} 
            }
	    \end{mathpar}
    \caption{Extension of the WebAssembly language with primitives.}
	    \label{fig:primitive-execution}
    \end{minipage}
\end{figure}

Although the debugger can be viewed as a single monolithic unit, it is more practical to conceptualize it as comprising three distinct components: the WebAssembly virtual machine (VM), which defines the foundational language semantics; the server, which operates as the remote debugger on the microcontroller; and the client running on a desktop computer, which interacts with the server to facilitate debugging. We begin by outlining the WebAssembly configuration, followed by a detailed discussion of the server configuration. Next, we examine the client configuration, and give an overview of the complete remote concolic debugger setup. In later sections, we delve into the semantics of both the server and the client components. 

\subsubsection{WebAssembly Configuration}
A core component of the WebAssembly configuration, shown in \cref{fig:configuration-remote-debugger} is the program state $K = \{ \rho,\updelta,st,\mu,e^* \}$ where $\rho$ are the local variables, $\updelta$ are the global variables, $st$ is the data stack, $\mu$ is the store and $e^*$ is the instruction stack.
The semantics of WebAssembly are defined as a small-step operational semantics, where the program state $K$ is updated through a transition relation $\hookrightarrow_{i}$. While the semantics of WebAssembly are quite intricate, it is sufficient to know that the transition relation $\hookrightarrow_{i}$ as defined in~\cite{haas17} is fully deterministic. 
On top of these deterministic base semantics we add nondeterministic primitive operations.

The WebAssembly VM used in this work supports a foreign function interface to define primitive operations. These primitives are exposed as a special WebAssembly module in the VM. By importing this module, programs can interact with external I/O elements, enabling functionality that cannot be achieved using standard \emph{deterministic} WebAssembly instructions. Since primitives interact with the external environment and perform operations beyond the scope of standard WebAssembly instructions, they are the \emph{only} source of non-determinism.

In the semantics, there are two primitive tables $P^{Out}$ and $P^{In}$ containing respectively the output and input primitives.
These tables are consulted to determine whether an instruction is invoking a primitive. 
We use $P^{*}$ to reference both tables as one. The notation `non-prim K' indicates that $K = \{ \rho,\updelta,v_a^* : st,\mu,call \; j \}$ and $P^{*}(j) = p$ does not hold, meaning the instruction is not a primitive. Conversely, the notation `prim K' indicates that $P^{*}(j) = p$ holds, meaning that the next instruction will be a primitive call. 
The function types show how input primitives return a single value $v$, and output primitives return no values.
\Cref{fig:primitive-execution} shows how input and output primitves are evaluated.
We use the notation $\lfloor p(v_0^*) \rfloor$ to indcate that the primitive function $p$ is executed outside the WebAssembly semantics.
For the input primitives the return value $v$ is nondeterministic as indicated by the notation $v \in \lfloor p(v_0^*) \rfloor$.

\subsubsection{Server Configuration}
\begin{figure}
	$$
	\begin{array}{ l l c l }
    \multicolumn{3}{l}{\textsf{\textbf{Server Configuration}}} & \\
		\emph{(Messages)}   & msg_s & \Coloneqq  & \mathit{step} \; | \; \mathit{pause} \; | \; \mathit{play} \;| \; \mathit{break}^{+} id \;| \; \mathit{break}^{-} id \; | \; \mathit{mock} \; v \;|  \\
                               &      &            & \mathit{inspect}  \; | \; \mathit{reset}\\
		\emph{(Breakpoints)}     & bps & \Coloneqq & \varnothing \; | \; bps \cup \{ id \}\\
		\emph{(Execution state)} & es & \Coloneqq  & \textsc{Running} \; | \; \textsc{Paused} \\
	   	\emph{(Server state)}  & S & \Coloneqq  & \langle es, msg_s^i, msg_c^o, bps, c_{instr}, K\rangle \\
	\end{array}
	$$
	\caption{Server configuration of the remote debugger.}
	\label{fig:server-configuration}
\end{figure}
\label{sec:server-configuration}
The server state $S$ shown in \cref{fig:server-configuration} is a tuple ${\langle es, msg_s^i, msg_c^o, bps, c_{instr}, K\rangle}$ where $es$ is the execution state, 
 $msg_s^i$ is the incoming message queue for receiving from the client, $msg_c^o$ is the outgoing message queue for sending to the client, $bps$ are the currently set breakpoints, $c_{instr}$ is a counter tracking how many instructions were executed since the last synchronisation with the client, and $K$ is the current program state.
The execution state indicates whether the program is currently \textsc{Running} or \textsc{Paused}.
Incoming messages from the client are server messages $msg_s$.
Breakpoints can control the execution of the program and are added or removed using the $break^{+} id$ and $break^{-} id$ messages respectively. 
In the following sections we explain how the $c_{instr}$ counter allows for efficiently synchronizing the client and server.

\subsubsection{Client Configuration}
\label{sec:client-configuration}
\begin{figure}
    \centering
    \begin{minipage}[b]{0.58\textwidth}
        \centering
        $$
	\begin{array}{ l l c l }
    \multicolumn{3}{l}{\textsf{\textbf{Client Configuration}}} & \\
		\emph{(Messages)}   & msg_c & \Coloneqq  & \mathit{prim}(c_{instr}, v) \; | \\
		&&&\mathit{executed}(c_{instr}) \; |  \\
		&&&\mathit{slide}(T_{target}) \; | \\
		&&&\mathit{snapshot}(K) \\
		\emph{(Tree)}   & T & \Coloneqq  & \mathit{Node} \; (tl, T)^* \\
     	        \emph{(Tree-Labels)}   & tl & \Coloneqq  & \mathit{step}  \; |  \; \mathit{mock} \; v      \\
	    \emph{(Client State)}  & C & \Coloneqq  & \langle msg^i_c, msg^o_s, T,T_{curr} \rangle \\
	\end{array}
	$$
	\caption{Client configuration of the remote debugger.}
	\label{fig:client-onfiguration}    \end{minipage}
    \hfill
    \begin{minipage}[b]{0.4\textwidth}
        \centering
        \begin{tikzpicture}[main/.style = {draw, circle, minimum size=0.5cm},scale=0.8]
    \begin{scope}[shift={(0, 0)}]
			\node[main,fill=blue!30] (1) at (0,0) {};
        \node[main,fill=blue!30] (2) at (2,0) {};
			\node[main] (4a) at (4,1) {};
			\node[main,fill=blue!30] (4b) at (4,-1) {};
			\node[main] (5a) at (6,1) {};
			\node[main,fill=blue!7] (5b) at (6, -0.5) {};
			\node[main] (5c) at (6, -1.5) {};
        \draw[->] (1) -- node[midway, above] {$\mathit{step}$} (2);
			\draw[->] (2) to [out=5,in=200] node[midway, above, yshift=4pt, xshift=-12pt] {$\mathit{mock} \; v_1$} (4a);
			\draw[->] (2) to [out=-5,in=160] node[midway, below, yshift=0pt, xshift=-12pt] {$\mathit{mock} \; v_2$} (4b);
        \draw[->] (4a) -- node[midway, above] {$\mathit{step}$} (5a);
        \draw[->] (4b) to [out=5,in=190] node[midway, above, yshift=3pt, xshift=-4pt] {$\mathit{mock} \; v_1$} (5b);
			\draw[->] (4b) to [out=-5,in=170] node[midway, below, yshift=-3pt, xshift=-4pt] {$\mathit{mock} \; v_2$} (5c);
			\draw[->] ++(0,0.6) -- (1);
			\node[] (A0) at (0, 0.9) {$T$};
			\draw[->] ++(4,-0.1) -- (4b);
			\node[] (A0) at (4, 0.2) {$T_{curr}$};
			\draw[->] ++(6,0.2) -- (5b);
			\node[] (A0) at (6, 0.5) {$T_{next}$};
    \end{scope}        \end{tikzpicture}
        \caption{Example of a multiverse tree. The edges in this tree are labeled using the debug operations needed to traverse from one node to another.}
	    \label{fig:timelines}
    \end{minipage}
\end{figure}

The client state shown in \cref{fig:client-onfiguration} is represented as a tuple $C = \langle msg^i_c, msg^o_s,$ $T, T_{curr} \rangle$ where $msg^i_c$ is the incoming message queue for receiving messages from the server, $msg^o_s$ is the outgoing message queue for sending messages to the server, $T$ is the root of the multiverse tree and $T_{curr}$ is the current node in the multiverse tree. 

The multiverse tree is used to represent the possible execution paths of the program and consists of deterministic and nondeterministic nodes. In \cref{fig:timelines} we show an example of a multiverse tree with a root node $T$ and a current node $T_{curr}$. In the semantics this tree would be described as $\mathit{Node} \; [(\mathit{step}, \mathit{Node} \; [(\mathit{mock} \; v_1, ...), \;$ $ (\mathit{mock} \; v_2, ...)])]$. This indicates there is a root $Node$ with one child connected using a $step$ edge. This child has two children each connected using a $mock$ edge. The edge labels in the multiverse tree indicate the debugging operations needed to traverse the tree. For example, to move from node $T_{curr}$ to $T_{next}$, the $mock \; v_1$ operation has to be performed. How each of these operations work exactly and how they affect the tree is discussed in \cref{sec:client-semantics}.

\subsection{Global Debugging Semantics}\label{sec:global-semantics}
\Cref{fig:serverClientCommunication} shows the semantics describing the communication between the server and the client in the debugger.
These rules also specify when the server can take a single step. 
Each rule is defined as a transition (denoted by $\rightarrow$) over the pair of client and server configurations $C|S$.

Whenever the server puts a message for the client in its outbox the \textsc{Server-To-Client} rule applies which causes the client to process this message using the client transition relation $\leadsto$.
The server-to-client communication has priority over all rules in the system, therefore on the client outgoing messages are always sent after processing all incoming messages.

When the server has no outgoing messages and the client has one or more messages in its outbox compatible ($\approx$) with the current state of the server, the message is processed by the server (denoted by $\hookrightarrow_{d,i}$) as shown in the \textsc{Client-To-Server} rule.
This compatibility relation essentially means the client can only send messages when the server is in the \textsc{Paused} state.
The only exception to this rule is the $\mathit{pause}$ message, which can only be sent from the to server when it is in the \textsc{Running} state\footnote{Note that in our semantics the user can easily bring the debugger in a stuck state by sending an incompatible message, we consider such scenario's to be part of a bug in the frontend of the debugger.}. This lets the client pause the server at any time.

When all message queues of the server and client are empty, the \textsc{Server-Step} rule applies and the server can take a single step under the $\hookrightarrow_{d,i}$ relation. In the following section we will focus on how this reduction relation is defined. 

Note that in our semantics, we have chosen to make all operations \emph{synchronous}, the server is blocked while the client processes messages. 
This decision is primarily motivated by clarity, as \emph{asynchronous} semantics would be more difficult to understand and reason about.
\begin{figure*}
	\begin{mathpar}
		\inferrule[(\textsc{Server-To-Client})]
		{
			\langle msg_c, msg_s^*, T,T_{curr} \rangle \leadsto C' \\
		}
		{
			\langle \varnothing, msg_s^*, T,T_{curr} \rangle \; | \; \langle es, \varnothing, msg_c, bps, c_{instr}, K \rangle
			\rightarrow
			C' \; | \; \langle es, \varnothing, \varnothing, bps, c_{instr}, K \rangle
		}\\
		
		\inferrule[(\textsc{Client-To-Server})]
		{
			S = \langle es, \varnothing, \varnothing, bps, c_{instr}, K \rangle \\
			(es \approx msg_s) \\
			\langle es, msg_s, \varnothing, bps, c_{instr}, K \rangle \hookrightarrow_{d,i} S' \\
		}
		{
			\langle \varnothing, msg_s . msg_s^*, T,T_{curr} \rangle \; | \; S
			\rightarrow
			\langle \varnothing, msg_s^*, T,T_{curr} \rangle \; | \; S'
		} 

		\inferrule[(\textsc{Server-Step})]
		{
			noMessages(C,S) \\ S \hookrightarrow_{d,i} S' \\
		}
		{
			C  \; | \; S
			\rightarrow
			C \; | \; S'
		}
	\end{mathpar}
	\caption{Global communication and debugging rules: server-to-client messages have priority over client-to-server messages. When all outboxes and inboxes are empty the client can take a single step. }
	\label{fig:serverClientCommunication}
\end{figure*}

\subsection{Server Semantics}
\label{sec:server-semantics}

Having discussed the global communication rules, we can now focus on local client and server reduction rules. We start our overview with the reduction relation of the sever $\hookrightarrow_{d,i}$.

The server has two states: \textsc{Running} or \textsc{Paused}, the rules for both cases are similar.
In the \textsc{Running} state, the server will execute the program while keeping track of a counter tracking the number of executed instructions since the last synchronisation with the client. Whenever an input primitive is executed, a so-called $\mathit{prim}(c_{instr}, v)$ message is sent to the client. This message contains the counter and the return value of the nondeterministic input primitive. After sending this message the server will reset the counter to $0$. The full rules for this behaviour can be found in \cref{fig:runServerRules} which is part of \cref{app:running-rules}.

The \textsc{Paused} rules are similar but instead always notify the client of every executed instruction so that the client knows exactly where in the program the server is.
The rules shown in \cref{fig:server:paused} define the semantics in the \textsc{Paused} state.
The server can go into the \textsc{Paused} state by hitting a breakpoint or because the developer sent a $pause$ message to the microcontroller. 
When switching to the \textsc{Paused} mode the client will be notified how many instructions were executed since the last synchronisation with the client.

Once the debugger is paused, the client can instruct the server to step through the execution by sending $step$  or $mock$ messages. 
When stepping, the server differentiates between normal instructions, calls to output primitives, and calls to input primitives as reflected in the \textsc{Dbg-Step}, \textsc{Dbg-Step-Prim-Out} and \textsc{Dbg-Step-Prim-In} rules.
When taking a step over a regular instruction or a call to an output primitive, the server notifies the client that one deterministic step was taken using $\mathit{executed}(1)$. When stepping over a call to a nondeterministic input primitive, the client is notified of the value returned using $\mathit{prim}(1,v)$. For the primitives $\lfloor p(v_a^*) \rfloor$ is used to indicate the operation is performed outside of the system. In this case the value $v$ produced outside of the system and can be nondeterministic.

The \textsc{Dbg-Mock} rule allows the developer to pick which execution should be explored by controlling the execution of nondeterministic primitives.
Instead of executing the primitive, the debugger replaces the arguments $v_a$ to the call with the supplied value $v$ and removes the $\text{call} \; e$ from the execution stack. An important requirement here is that $v$ must be a value that the primitive can produce in the underlying language semantics.
With this rule alone, a developer has to manually choose the value to be explored. However, it is often complicated to figure out which values are needed to explore specific execution paths such as in the basic example in the previous section. In \cref{sec:concolic}, we explain how we automatically generate the necessary values removing the need to enumerate countless options.

The \textsc{Dbg-Play}, \textsc{Dbg-Inspect} and \textsc{Dbg-Reset} rules of  the debugger, detailed in \cref{fig:additional-paused}, respectively transition the debugger to the running state, request a snapshot and restart execution from the start of the program.

\begin{figure*}
	\begin{mathpar}

		\inferrule[(\textsc{Dbg-Step})]
       		{
				\textsf{non-prim} \; K \\
				K \hookrightarrow_{i} K' \\
            }
            {
				\langle \textsc{Paused}, \mathit{step}, \varnothing, bps, 0, K \rangle
				\hookrightarrow_{d,i} 
				\langle \textsc{Paused}, \varnothing, \mathit{executed}(1), bps, 0, K' \rangle
			}\\
		\inferrule[(\textsc{Dbg-Step-Prim-Out})]
       		{
				K = \{ \rho,\updelta,v_a^*:v^*,\mu,e:e^* \} \\
				e = \textbf{call} \; j \\
				P^{Out}(j) = p \\
				\lfloor p(v_a^*) \rfloor
            }
            {
				\langle \textsc{Paused}, \mathit{step}, \varnothing, bps, 0, K \rangle
				\hookrightarrow_{d,i} 
				\langle \textsc{Paused}, \varnothing, \mathit{executed}(1), bps, 0, \{ \rho,\updelta,v^*,\mu,e^* \} \rangle
			}\\

		\inferrule[(\textsc{Dbg-Step-Prim-In})]
       		{
				K = \{ \rho,\updelta,v_a^*:v^*,\mu,e:e^* \} \\
				e = \textbf{call} \; j \\
				P^{In}(j) = p \\
				v \in \lfloor p(v_a^*) \rfloor
            }
            {
				\langle \textsc{Paused}, \mathit{step}, \varnothing, bps, 0, K \rangle
				\hookrightarrow_{d,i} 
				\langle \textsc{Paused}, \varnothing, \mathit{prim}(1, v), bps, 0, \{ \rho,\updelta,v:v^*,\mu,e^* \} \rangle
			}\\
		\inferrule[(\textsc{Dbg-Mock})]
       		{
				K = \{ \rho,\updelta,v_a^*:v^*,\mu,e:e^* \} \\
				e = \textbf{call} \; j \\
				P^{In}(j) = p \\
				v \in codom(p)
            }
            {
				\langle \textsc{Paused}, \mathit{mock} \; v, \varnothing, bps, 0, K \rangle
				\hookrightarrow_{d,i}
				\langle \textsc{Paused}, \varnothing, \mathit{prim}(1, v), bps, 0, \{ \rho,\updelta,v:v^*,\mu,e^* \} \rangle
			}
		\end{mathpar}		
	\caption{Server rules used when the debugger is paused.}
	\label{fig:server:paused}
\end{figure*}

\subsection{Client Semantics}
\label{sec:client-semantics}

The client in this system does most of the heavy lifting since it does not run on the microcontroller.
Instead the client has a message queue that processes messages received from the microcontroller (server). 
Using these messages the debugger constructs the navigable multiverse tree. 
The tree has the structure described in \cref{fig:configuration-remote-debugger} as previously explained in \cref{sec:client-configuration}. Every node in the tree can have zero or more children. The edges of the tree are marked with debugging operations used to navigate the tree. For deterministic nodes the $step$ operation is used, for nondeterministic nodes $mock$ will be used. This tree and the current node is stored alongside the incoming and outgoing messages in the debugger state.

\subsubsection{Building and Navigating the Tree}
During the execution of a program, the multiverse tree stored in the client is both built up and navigated by the debugger. 
Two types of messages notify the client about executed instructions on the server: $\mathit{executed}(c_{instr})$ and $\mathit{prim}(c_{instr}, v)$. 
The former indicates the server executed a specified number of deterministic instructions, denoted by $c_{instr}$. 
The latter notifies the client that the server has executed $c_{instr}$ number of deterministic instructions followed by a nondeterministic primitive that returned $v$.
When these messages are sent by the server they are transferred into the inbox of the client. 
Upon arrival at the client, the debugger processes these message with the \textsc{Executed} and \textsc{Prim} rules shown in \cref{fig:client-rules}. 

Both the \textsc{Executed} and \textsc{Prim} rules make use of the $\Rightarrow_{Traverse}^*$ operation. This operation navigates the tree and creates new paths if necessary based on the provided sequence of input messages. In essence it just follows the messages which are also the labels of the edges in the tree one by one and creates new edges if the current message it is processing does not exist yet. The exact details for these rules can be found in \cref{appendix:traversal-rules}.

\begin{figure}
	\centering
    \begin{tikzpicture}[main/.style = {draw, circle, minimum size=0.5cm},scale=0.875]
    \begin{scope}[shift={(0, 0)}]
		\node[main,fill=blue!30] (1) at (0,0) {};
        \node[main,fill=blue!30] (2) at (2,0) {};
        \node[main,fill=blue!7] (3) at (4,0) {};
        \node[main,fill=blue!7] (4) at (6,0) {};
        \draw[->] (1) -- node[midway, above] {$step$} (2);
        \draw[->,line width=0.25mm,blue!50] (2) -- node[midway, above] {$step$} (3);
        \draw[->,line width=0.25mm,blue!50] (3) -- node[midway, above] {$step$} (4);
        \draw[->] ++(2, 0.6) -- (2);
		\node[] (A0) at (2, 0.9) {$T_{curr}$};
		\draw[->] ++(6, 0.6) -- (4);
		\node[] (A0) at (6, 0.9) {$T_{curr}'$};
    \end{scope}        \end{tikzpicture}
    \caption{Example showing the effect of the $executed(2)$ message on the multiverse tree.}
    \label{fig:executed-rule}
\end{figure}

\begin{figure}
    \centering
    \begin{minipage}[t][5cm]{0.48\textwidth}
        \centering
        \begin{tikzpicture}[main/.style = {draw, circle, minimum size=0.5cm},scale=0.875]
		\begin{scope}[shift={(0, 0)}]
    		\node[main,fill=blue!30] (1) at (0,0) {};
            \node[main,fill=blue!30] (2) at (2,0) {};
            \node[main] (3) at (4,0) {};
            \node[main] (4) at (6,0) {};
            \draw[->] (1) -- node[midway, above] {$\mathit{step}$} (2);
            \draw[->] (2) -- node[midway, above] {$\mathit{step}$} (3);
            \draw[->] (3) -- node[midway, above] {$\mathit{mock} \; v_1$} (4);
            \draw[->] ++(2, 0.6) -- (2);
    		\node[] (A0) at (2, 0.9) {$T_{curr}$};
			\node[anchor=west] (A2) at (6.5, 0) {Before};
        \end{scope}
        \begin{scope}[shift={(0, -1.5)}]
    		\node[main,fill=blue!30] (1) at (0,0) {};
            \node[main,fill=blue!30] (2) at (2,0) {};
            \node[main,fill=blue!7] (3) at (4,0) {};
            \node[main] (4a) at (6, 0.5) {};
            \node[main,fill=blue!7] (4b) at (6, -0.5) {};
            \draw[->] (1) -- node[midway, above] {$\mathit{step}$} (2);
            \draw[->,line width=0.25mm] (2) -- node[midway, above] {$\mathit{step}$} (3);
            \draw[->] (3) to [out=5,in=190] node[midway, above, yshift=4pt, xshift=-2pt] {$\mathit{mock} \; v_1$} (4a);
            \draw[->,blue!50,line width=0.25mm] (3) to [out=5,in=170] node[midway, below, yshift=-4pt, xshift=-2pt] {$\mathit{mock} \; v_2$} (4b);
            \draw[->] ++(2, 0.6) -- (2);
    		\node[] (A0) at (2, 0.9) {$T_{curr}$};
    		\draw[->] ++(6, -1.1) -- (4b);
    		\node[] (A1) at (6, -1.3) {$T_{curr}'$};
			\node[anchor=west] (A2) at (6.5, 0) {After};
        \end{scope}
    \end{tikzpicture}
    \caption{Example showing the effect of the $prim(2, v_2)$ message on a tree with pre-existing edges. In this case there is one $step$ edge present, which is followed without any change. Next, a new node is connected using a $mock \; v_2$ edge.}
    \label{fig:prim-rule-pre-existing}
    \end{minipage}
    \hfill
    \begin{minipage}[t][5cm]{0.48\textwidth}
        \centering
        \begin{tikzpicture}[main/.style = {draw, circle, minimum size=0.5cm},scale=0.875]
        \begin{scope}[shift={(0, 0)}]
    		\node[main,fill=blue!7] (1) at (0,0) {};
            \node[main,fill=blue!7] (2) at (2,0) {};
            \node[main,fill=blue!30] (3a) at (4, 0.5) {};
            \node[main,fill=blue!7] (3b) at (4, -0.5) {};
            \node[main] (4a) at (6, 0.5) {};
            \node[main,fill=blue!30] (4b) at (6, -0.5) {};
            \draw[->,blue!50,line width=0.25mm] (1) -- node[midway, above] {$\mathit{step}$} (2);
            \draw[->] (2) to [out=5,in=190] node[midway, above, yshift=3pt, xshift=-4pt] {$\mathit{mock} \; v_1$} (3a);
            \draw[->,blue!50,line width=0.25mm] (2) to [out=5,in=170] node[midway, below, yshift=-3pt, xshift=-4pt] {$\mathit{mock} \; v_2$} (3b);
			\draw[->] (3a) -- node[midway, above] {$\mathit{step}$} (4a);
			\draw[->,blue!50,line width=0.25mm] (3b) -- node[midway, above] {$\mathit{step}$} (4b);
            \draw[->,blue!50,line width=0.25mm] (3a) to [out=150,in=40] node[midway, above, yshift=1pt, xshift=0pt] {$\mathit{reset}$} (1);
            \draw[->] ++(4, 1.1) -- (3a);
    		\node[] (A0) at (4, 1.4) {$T_{curr}$};
            \draw[->] ++(6, -1.1) -- (4b);
    		\node[] (A0) at (6, -1.3) {$T_{target}$};
        \end{scope}
    \end{tikzpicture}
    \caption{Example showing how to slide to a target that is not a descendant of $T_{curr}$. In this scenario, the server will first be reset and will then, from the start of the program follow the path to the destination node.}
    \label{fig:slide-reset}
    \end{minipage}
\end{figure}

\subparagraph*{Executed}
When $\mathit{executed}(c_{instr})$ arrives the \textsc{Executed} rule will take $c_{instr}$ steps forward in the tree starting from the current node $T_{curr}$.
It does this by using $c_{instr}$ $step$ operations as the sequence of operations. The $\mathit{executed}(c_{instr})$ message specifies that the executed path of $c_{instr}$ was fully deterministic, as a result the nodes that are traversed will only have one or zero children. If the node has a child it becomes the next node, if it does not, a new node is created and that node becomes the next node. In \cref{fig:executed-rule} we show the effect of receiving an $\mathit{executed}(2)$ message in a very basic tree. If $T_{curr}$ did not have any children, two new nodes will be attached an the last node becomes the new $T_{curr}$. If $T_{curr}$ did already have a child, the edge will simply be followed. For the second edge the same rule applies.

\subparagraph*{Prim}
Upon receiving $\mathit{prim}(c_{instr}, v)$ the client learns that $c_{instr}$ instructions were executed ending with a nondeterministic operation that returned value $v$. The new path consists of $c_{instr} - 1$ $\mathit{step}$ operations and one $\mathit{mock} \; v$ operation. Existing nodes are traversed, and new nodes are created when needed. In the case of $\mathit{mock} \; v$ each of the start node's existing children could have a different message on their edge. In this case the debugger will add a new child to this node and make that child the next node. This is illustrated in \cref{fig:prim-rule-pre-existing}. In this example there are pre-existing $\mathit{step}$ and $\mathit{mock} \; v_1$ edges but there is no $\mathit{mock} \; v_2$ edge.

\begin{figure}
	\begin{mathpar}
		\inferrule[(\textsc{Executed})]
		{
			\langle \mathit{step}^{c_{instr}}, T, T_{curr}  \rangle \Rightarrow_{Traverse}^*  \langle \varnothing, T', T_{curr}' \rangle
		}
		{
			\langle \mathit{executed}(c_{instr}) , \varnothing, T,T_{curr} \rangle
			\leadsto
			\langle \varnothing , \varnothing, T',T_{curr}' \rangle
		}\\
		\inferrule[(\textsc{Prim})]
		{
			\langle \mathit{step}^{c_{instr} - 1} : \mathit{mock} \; v, T, T_{curr} \rangle \Rightarrow_{Traverse}^*   \langle \varnothing, T', T_{curr}' \rangle
		}
		{
			\langle \mathit{prim}(c_{instr}, v) , \varnothing, T,T_{curr} \rangle
			\leadsto
			\langle \varnothing, \varnothing, T',T_{curr}' \rangle
		}\\

			\inferrule[(\textsc{Slide})]
       		{
				T_{curr} \vdash^* T_{target} \\
				T_{curr}, \varnothing \Rightarrow_{Path}^* T_{target}, msg^* \\
            }
            {
				\langle \mathit{slide} \; T_{target}, \varnothing, T,T_{curr} \rangle
				\leadsto
				\langle \varnothing, msg^*, T,T_{curr} \rangle
			}\\
			\inferrule[(\textsc{Slide-Restart})]
       		{
				\lnot (T_{curr} \vdash^* T_{target}) \\
				T, \varnothing \Rightarrow_{Path}^* T_{target}, msg^* \\
            }
            {
				\langle \mathit{slide} \; T_{target}, \varnothing, T,T_{curr} \rangle
				\leadsto
				\langle \varnothing, reset : msg^*, T,T \rangle
			}\\
			\inferrule[(\textsc{Concolic-Receive})]
			{
				s(K) = K_{cs} \\
				\mathit{concolic}(K_{cs}) = T_{conc} \\
				T' = \mathit{replace}(T, T_{curr}, T_{conc}) \\
            }
            {
				\langle \mathit{snapshot}(K), \varnothing, T,T_{curr} \rangle
				\leadsto
				\langle \varnothing, \varnothing, T',T_{conc} \rangle
			}
	\end{mathpar}
	\caption{Client-Side Multiverse Debugging operations.}
	\label{fig:client-rules}
\end{figure}

\subparagraph*{Slide}
Using the previous rules, the tree is automatically extended during the execution of a program and new branches can be added with the $mock$ operation. To easily navigate the multiverse tree users can move 
to any pre-existing state using the \textit{slide}\footnote{Slide is named after the popular 1995 sciencefiction TV series `Sliders', in which the main characters slide into a parallel universe each episode.} operation.

The slide operation consists of two rules. Either the user slides to a target state downstream of the current execution path (indicated by $T_{curr} \vdash^* T_{target}$). Or they slide into a state that cannot be reached by going forward in the current execution. When this happens the debugger will instruct the VM to restart the program and execute the selected path as shown in \cref{fig:slide-reset}. This approach differs from MIO~\cite{mio}, as MIO is able to step back in time instead of restarting the execution. In both rules the debugger will have to determine the necessary debugging operations needed to go from the current state (or the start state in case of \textsc{Slide-Restart}) to the target state. Thanks to the structure of the multiverse tree, this is simply the list of edge labels on the path to the target state. To obtain this list the $\Rightarrow_{Path}$ rule is used in the semantics. The exact definition of this rule can be found in \cref{appendix:path-rules}.

Interestingly, the slide rules do not directly change the current node in the tree to $T_{target}$.
Instead, they instruct the remote debugger to move to a particular location in the tree.
When doing so, the server will notify the client of any executed instructions and primitives using the $executed$ and $prim$ messages, which will in turn change the current node in the tree.
Because \textsc{Slide-Restart} resets the execution, it first changes the current node to the root node $T$.
The $\mathit{executed}$ and $\mathit{prim}$ messages then walk forward again in the tree.

\subparagraph*{Concolic-Receive} Finally, the tree can also be automatically extended with paths that result in full code coverage using \textsc{Concolic-Receive}. This is the topic of the next section.

\section{Suggesting Interesting Paths using Concolic Analysis}\label{sec:concolic}
\begin{figure}
	$$
	\begin{array}{ l l c l }
    	\emph{(Symbolic expression)}   & \hat{s}                    & \Coloneqq  & v \; | \; \hat{x} \; | \;  \hat{t.unop}(\hat{s}) \; | \; \hat{t.binop}(\hat{s}, \hat{s}) \; | \; ... \\
		\emph{(Symbolic environment)}  & \upvarepsilon              & \Coloneqq & \hat{x} \rightarrow v \\
		\emph{(Path condition)}         & \pi                        & \Coloneqq & \hat{s} \\
		\emph{(Concolic program state)}        & K_{cs} & \Coloneqq & \langle K,\hat{\rho},\hat{\updelta},\hat{st},\hat{\mu},\upvarepsilon,\pi \rangle \\
		\emph{(Concolic Tree)}         & \hat{T}                     & \Coloneqq & \mathit{Node} \; (tl, \hat{T})^* \; (\pi, \upvarepsilon)^* \\
    \end{array}
	$$
	\caption{Configuration of the concolic execution engine.}
	\label{fig:concolic-config}
\end{figure}

The $\mathit{mock}$ and $\mathit{slide}$ operations already allow the multiverse tree to be extended and navigated. That said, to find bugs the developer still needs to enumerate all possible sensor values. Yet many values result in the same execution path. Therefore, we propose using concolic execution to identify a subset of sensor values that provide full control flow coverage.

Concolic execution~\cite{godefroid05, king76} is a technique mostly used for automated testing and fuzzing~\cite{baldoniSurveySymbolicExecution2019} that generates inputs for programs so that each execution path of a program can be explored. 
Concolic execution is an iterative algorithm that explores a new execution path in every iteration. The algorithm derives a model, i.e. assignment of symbolic variables which satisfies a path condition.  In the first iteration, symbolic variables are given an arbitrary value. Using these values the program is executed concolically. While this happens the concolic executor will build up a path condition. At the end of this execution the algorithm will typically use an SMT solver to find a path different from the already explored paths. To do so the algorithm asks the SMT solver to solve the equation representing conjunction of the negations of all explored path conditions. Using this model a new iteration is started again until no new model can be found. An illustrative example of this algorithm can be found in \cref{appendix:example-concolic-execution}.

Here we show how the ideas of this technique can be used to generate not just a set of inputs for a program but a multiverse tree containing all future execution paths that could be explored starting from the current execution state.

\subsection{Concolic Execution in Debugging}

In order to integrate concolic execution into a live debugging session, we cannot follow the usual strategies of current concolic execution engines. 
Firstly, our analysis begins not from the start of the program or an existing concolic state (as in online concolic execution) but instead starts from the concrete execution state of a program running on a microcontroller.

Secondly, we are working with embedded systems where concolic execution cannot run locally.  As a result, the analysis must be performed remotely. 
To initiate the analysis from the current execution state of the microcontroller, the state needs to be transferred to the client performing the concolic execution.
The custom virtual machine enables this state transfer by sending an $inspect$ message to the debugging server, as previously described. 
When receiving this message the server serializes its state and transfers it to the client. However, this is purely a concrete state without any symbolic components. To start the analysis from this concrete state, the debugger first needs to transform it into a so-called concolic state in which every concrete element has a symbolic counterpart. 
Using the newly created concolic state, the analysis can start from the current execution state in the debugger.

For our debugger, we developed our own concolic execution engine for WebAssembly based on the semantics described in~\cite{marques22}.
To ensure that our theoretical model aligns with our practical implementation, which leverages an existing concrete execution engine during concolic execution, we adjusted the formalization. By doing so the concolic rules re-use the concrete $\hookrightarrow_i$ rules just like in the implementation.
This approach results in both server and client VMs using the same codebase for concrete execution. 
However, on the server side, we compile the code in a way that excludes the concolic execution functionality.

\subsection{Configuration of the Concolic Executor}
To start we describe the components of this execution engine shown in \cref{fig:concolic-config}. The core configuration for the concolic semantics is very similar to the core of the previously described debugger semantics. 
It consists of a collection of locals, globals, a stack and memory. These are part of the concrete state $K$. Unlike the concrete execution, each of these components now has a symbolic counterpart that we denote using $\hat{y}$ for a component $y$. The symbolic locals ($\hat{\rho}$), globals ($\hat{\updelta}$), stack ($\hat{st}$) and memory ($\hat{\mu}$) are all defined similarly to their concrete counterparts but now map an index to a symbolic expression instead of a concrete value. We previously defined concrete values as $v$, we now define symbolic expressions as $\hat{s}$. A symbolic expression can be a concrete value, a symbolic variable $\hat{x}$ or a unary, binary or tertiary operation on symbolic expressions. These are symbolic expressions and not values because it is not always possible to calculate the result of an expression when working symbolically.

Alongside the locals, globals, stack and memory the concolic execution engine also has some properties specifically needed for concolic execution. These are the symbolic environment ($\upvarepsilon$) and the path condition ($\pi$) stored in a concolic program state defined by $K_{cs}$.

Outside of the concolic state our concolic execution engine also keeps track of a multiverse tree which will later be added to the tree in the debugger. This tree denoted with $\hat{T}$ is very similar to the tree described in \cref{sec:remote-concolic-multiverse-debugging} but now has an additional component which we call the `$pathModels$'. This is a list of path conditions and their associated symbolic environment, denoted by $(\pi, \upvarepsilon)^*$. This additional information is used by the analysis to continually extend the current multiverse tree with new paths while merging paths that share a common prefix.

\subsection{Concolic Execution Semantics}
\label{sec:concolic-rules}

Using the previously defined components, we now define the concolic semantics in \cref{fig:concolic-semantics} (defined by $\Rightarrow_{cs}$). 
For brevity we do not list all rules here but give a representative subset.

The concolic semantics define a transition relation between two concolic states $K_{cs}$. 
All rules follow the following structure: $\langle K,  \hat{\rho},\hat{\updelta},\hat{st},\hat{\mu},\upvarepsilon, \pi   \rangle \Rightarrow_{cs} \langle K',  \hat{\rho}',\hat{\updelta}',\hat{st}',\hat{\mu}',\upvarepsilon', \pi'   \rangle$.

We first describe the rule used for simple binary arithmetic instructions, these instructions e.g. $i32.add$ simply take two elements from the stack, apply an operation to it, for example addition and place the result on the stack. In our concolic execution engine this is done by using the underlying concrete semantics to manipulate the concrete stack and applying the symbolic binary operation on two elements of the symbolic stack. In this case symbolic addition. This way we do not need to reimplement the concrete part of our interpreter for every instruction but can instead reuse the existing implementation.

\begin{figure}
	\begin{mathpar}
	\inferrule[(\textsc{Binop})]
       		{ 
		K =  \{ \rho,\updelta,st,\mu, t.binop : e^* \}  \\
		K  \hookrightarrow_i  K' \\
		\hat{st} = \hat{s_a} : \hat{s_b} : \hat{s}^* \\
     		\hat{st'} = \hat{t.binop}(\hat{s_a}, \hat{s_b}) : \hat{s}^*\\
            }
            {				
				\langle K,  \hat{\rho},\hat{\updelta},\hat{st},\hat{\mu},\upvarepsilon, \pi   \rangle
				\Rightarrow_{cs}
				\langle K',  \hat{\rho},\hat{\updelta},\hat{st'},\hat{\mu},\upvarepsilon, \pi   \rangle
			}
			\\
	\inferrule[(\textsc{Symbolic-Prim})]
       		{ 
		K =  \{ \rho,\updelta,v_a^* : st,\mu, \textbf{call} \; j : e^* \} \\
		K' =  \{ \rho,\updelta,v : st,\mu, e^* \} \\
		P^{In}(j) = p \\
                 \hat{x_d} = id(\textbf{call}\: j)    \\
		\hat{x_d} \in dom(\upvarepsilon) \\
		\upvarepsilon(\hat{x_d}) = v \\
            }
            {				
				\langle K,  \hat{\rho},\hat{\updelta},\hat{s}_a^* :\hat{st},\hat{\mu}, \upvarepsilon, \pi   \rangle
				\Rightarrow_{cs}
				\langle K',  \hat{\rho},\hat{\updelta},\hat{x_d} : \hat{st},\hat{\mu},\upvarepsilon, \pi   \rangle
			}
			\\
    \inferrule[(\textsc{Symbolic-Prim-Fresh})]
       		{
		K =  \{ \rho,\updelta,v_a^* : st,\mu, \textbf{call} \; j : e^* \} \\
		K' =  \{ \rho,\updelta,v : st,\mu, e^* \} \\
		P^{In}(j) = p \\
                 \hat{x_d} = id(\textbf{call}\: j)    \\
		\hat{x_d} \not\in dom(\upvarepsilon) \\
		v \in codom(p) \\
            }
            {				
				\langle K,  \hat{\rho},\hat{\updelta},\hat{s}_a^* :\hat{st},\hat{\mu},\upvarepsilon, \pi   \rangle
				\Rightarrow_{cs}
				\langle K',  \hat{\rho},\hat{\updelta},\hat{x_d} : \hat{st},\hat{\mu}, \upvarepsilon[\hat{x_d} \mapsto v], \pi   \rangle
			} \\
		\inferrule[(\textsc{If-True})]
       		{ 
				K =  \{ \rho,\updelta,v: st',\mu,  \textbf{if} \; tf \;  e_1^* \; \textbf{else} \; e_2^*  \textbf{ end} \}  \\
				K \hookrightarrow_i K' \\
            	\hat{st} = \hat{s}: \hat{st'}\\
				v \neq 0 \\
				\pi' = \pi \land (\hat{s} \neq 0)
            }
            {
				\langle K,  \hat{\rho},\hat{\updelta},\hat{st},\hat{\mu},\upvarepsilon, \pi   \rangle
				\Rightarrow_{cs}
				\langle K',  \hat{\rho},\hat{\updelta},\hat{st'},\hat{\mu},\upvarepsilon, \pi'   \rangle
			}\\
		\inferrule[(\textsc{If-False})]
       		{ 
				K =  \{ \rho,\updelta,v: st',\mu, \textbf{if } tf \;   e_1^* \textbf{ else } e_2^* \textbf{ end} \}  \\
				K \hookrightarrow_i K' \\
            	\hat{st} = \hat{s}: \hat{st'}\\
				v = 0 \\
				\pi' = \pi \land (\hat{s} = 0)
            }
            {
				\langle K,  \hat{\rho},\hat{\updelta},\hat{st},\hat{\mu},\upvarepsilon, \pi   \rangle
				\Rightarrow_{cs}
				\langle K',  \hat{\rho},\hat{\updelta},\hat{st'},\hat{\mu},\upvarepsilon, \pi'   \rangle
			}\\
		\end{mathpar}
	\caption{Subset of rules illustrating the operation of the concolic execution engine.}
	\label{fig:concolic-semantics}
\end{figure}

 \textsc{Symbolic-Prim-Fresh} is responsible for creating new symbolic variables. 
When an input primitive is called and the unique symbolic variable $\hat{x}_d$ associated with this call is not in the symbolic environment $\upvarepsilon$, a new symbolic variable with that name is created and is given an arbitrary concrete value $v$.  Note that we assume here that the $id$ function returns a different identifier for different uses of the same call. This symbolic variable is placed in $\upvarepsilon$.

\textsc{Symbolic-Prim} applies when the symbolic variable $\hat{x}_d$ already exists in the symbolic environment.
 In this scenario, the concrete value is read from $\upvarepsilon$. This concrete value is pushed on the concrete stack. The symbolic variable itself pushed on the symbolic stack.

The \textsc{If-True} rule is essential for building up many of the path conditions. Given that the current instruction is $ \textbf{if} \; tf \;  e_1^* \; \textbf{else} \; e_2^* \textbf{ end} $ and the condition on the stack is not $0$, the next instructions will be $e_1^*$. The concrete part of the state is updated in the same way as normal using the concrete semantics $\hookrightarrow_i$. The symbolic part of the state is updated by removing the expression representing the condition $\hat{s}$ from the stack and updating the path condition to show that the condition $\hat{s}$ is true.

The \textsc{If-False} rule works analogously except that the condition is false now and the branch will not be taken. Instead the instructions in the else case will be executed. The symbolic stack and path condition are updated accordingly to reflect that the condition is not true in this case.

\subsection{Generating Multiverse Trees}
A traditional concolic execution engine generates a series of models, each representing a unique execution path within the program. 
However, a multiverse debugger enables interactive exploration of these executions through a graph-based interface. 
For concolic multiverse debugging, the analysis must not produce a sequence of isolated models but instead construct a tree in which overlapping parts of different executions are merged.

This is what happens in the \textsc{Concolic-Receive} rule for the client shown in \cref{fig:client-rules}. It will, upon receiving a snapshot by first sending an $inspect$ message, take a concrete state and convert it into a concolic state using $s(K)$ (details in \cref{app:expansion-rules}). Then it uses the $\mathit{concolic}$ function described below to generate a tree of future execution paths. Finally, the resulting tree is attached to the current multiverse tree by replacing the current node with the root of the newly generated tree. 

The tree of future execution paths is built iteratively in the $\mathit{concolic}$ function described in \cref{alg:concolic}. 
To achieve this, the main loop of the analysis iteratively expands the multiverse tree after each concolic iteration using the $extendTree$ function.

\begin{algorithm}[t]
\caption{Concolic interpreter main loop}\label{alg:concolic}
\SetKwProg{Fn}{function}{}{}
\Fn{$concolic(K_{cs} = \langle K, \hat{\rho},\hat{\updelta},\hat{st},\hat{\mu}, \varepsilon,\pi \rangle)$}{
	$\Pi \gets true$\\
	$root \gets Node \; \varnothing \; \varnothing$\\
	
	\While{$\Pi \; is \; SAT \land belowLimit()$}{
	  $\upvarepsilon \gets Model(\Pi)$\\
	  $\langle K,\hat{\rho},\hat{\updelta},\hat{st},\hat{\mu}, \upvarepsilon, \pi  \rangle \Rightarrow_{cs}^n  \langle K',\hat{\rho}',\hat{\updelta}',\hat{st}',\hat{\mu}', \upvarepsilon', \pi' \rangle$\\
	  $root \gets \bf{extendTree}$$(root, \upvarepsilon', \pi',K)$\\
	  $\Pi \gets \Pi \land \lnot \pi'$\\
	}
	\Return withoutPathModels(root)
}

\Fn{$extendTree(\hat{T}_{curr}, \upvarepsilon, \pi, K)$}{
	$d \gets 0$ \\
	\While{$d < len(\upvarepsilon)$}{
		$\hat{T}_{prev} \gets \hat{T}_{curr}$ \\
		$K \hookrightarrow_{i} K'$\\
		\eIf{$\textsf{non-prim} \; K$}{
			$\hat{T}_{curr} \gets \hat{T}_{curr}.\mathit{attachOrFollow}(\mathit{step})$
		}{
			\eIf{$\mathit{findFirst} \; (\pi', \upvarepsilon') \; in \; \hat{T}_{curr}.pathModels$ $where \; Equivalent(\pi,\upvarepsilon,\pi',\upvarepsilon', d)$}{
				$\hat{T}_{curr} \gets \hat{T}_{curr}.\mathit{follow}(\mathit{mock} \; \upvarepsilon'[x_d])$
    		}{
				$\hat{T}_{curr} \gets \hat{T}_{curr}.\mathit{attach}(\mathit{mock} \; \upvarepsilon[x_d])$
    		}
			$d \gets d + 1$
		}
		$K \gets K'$ \\
		$\hat{T}_{prev}.pathModels.\mathit{add}((\pi, \upvarepsilon))$	
	}
	\Return $\hat{T}_{curr}$\\
}

\Fn{$Equivalent(\pi,\upvarepsilon,\pi',\upvarepsilon', d)$}{
	\Return $\pi(\upvarepsilon[x_{0..d} \mapsto \upvarepsilon'[x_{0..d}]]) \land$
	$\pi'(\upvarepsilon'[x_{0..d} \mapsto \upvarepsilon[x_{0..d}]])$
}
\end{algorithm}

To expand the existing tree the different symbolic values stored in the model $\upvarepsilon$ are used to walk through the existing multiverse tree. When walking through the tree, the algorithm adds or follows a path for every executed instruction. The pseudocode does this using the $\hookrightarrow_i$ relation\footnote{We use $\hookrightarrow_i$ here for simplicity, in the actual implementation the number of deterministic instructions between each choice point is simply stored in a counter, so there is no need to re-execute the program.}. If the instruction is deterministic, a node connected using a $step$ edge is added or followed. These edges never branch and simply indicate that there is a deterministic instruction that can be stepped over. If the edge does not exist yet, it is created, if it does exist, the edge is followed. If an instruction is nondeterministic, the algorithm will check if the path for this value and path condition is unique or not. If it is not unique, the algorithm will keep following the existing path. If it is unique, a new branch in the tree will be created.

To check if a path is unique with respect to a particular value, it is important to realise that a new path found by a concolic execution engine will
always differ from an existing path at some point. As such the tree building algorithm looks for the point where the new model's path condition branches off from one of the models in the current path (denoted by $\hat{T}_{curr}.pathModels$). This is the point at which the common prefix of two executions end.

To achieve this the algorithm will do a pairwise comparison between the models for each symbolic variable $x_d$ until a point where $x_d$ diverges. The algorithm does this by placing the first $d$ variables of the existing model in the new path condition and vice-versa. At some point one (or both) of these path conditions will no longer hold. At this point the two models diverge. From this point on the algorithm will attach new nodes to the tree, creating a new path, instead of following an existing path.

\subparagraph*{Tree construction example} To illustrate this algorithm, we will show how a multiverse tree can be built for the example program shown in \cref{lst:concolic-merging} which uses a simple loop with an if statement. For this program, we will give a few example models that can be found by concolic execution and how they would be added into the multiverse tree. 

\begin{figure}[]
\begin{lstlisting}[language=C, style=CStyle,escapechar=']
for (int i = 0; i < 3; i++) {
    if (chip_analog_read(sensorPin) < 5) {
        ...
    }
}
\end{lstlisting}
\caption{Example program to illustrate the operation of path merging.}
\label{lst:concolic-merging}
\end{figure}

The first concolic iteration could result in the model $\{x_0=4,x_1=4,x_2=4\} $ with path condition $\pi = x_0 < 5 \land x_1 < 5 \land x_2 < 5$, shown at the top of \cref{fig:building-the-tree}. After each concolic iteration, the algorithm will run the $extendTree$ method. At the start of the algorithm the root node has no children and has no associated models. The $extendTree$ function will check if there already exists a model that has a common prefix with the current model up to depth $d$. In every loop iteration this $d$ is increased, until the divergence point. As there are no models associated with the current node, it will result in a new path. When creating this path the model and path condition are also added into the $pathModels$ of the nodes.

The second concolic iteration then finds a different path with model $\{x_0=3,x_1=5,x_2=4\}$ and path condition $  \pi' = x_0 < 5 \land \lnot (x_1 < 5) \land x_2 < 5$. This model is shown below the first model in \cref{fig:building-the-tree}.
The $extendTree$ function will now follow the existing path up to the divergence point, at which point it will start a new branch by attaching new nodes.

Initially $d$ is $0$, to check if the branches diverge the algorithm will substitute $x_0$ from the current tree into the new path condition and vice-versa as shown in the $Equivalent$ function. If one of the path conditions does not hold using this different $x_0$, the branch has diverged. In this case $3 < 5 \land 4 < 5 \land 4 < 5$ holds and $4 < 5 \land \neg(5 < 5) \land 4 < 5$ also holds. This means the value of $x_0$ is valid for both path conditions which indicates these paths share a common prefix. Consequently, the algorithm can keep following the existing path. Then $d$ increases to $1$. At this point the algorithm will substitute both $x_0$ and $x_1$ from the original model into the new path condition and vice-versa. This results in $3 < 5 \land 5 < 5 \land 4 < 5 = 5 < 5$ and $4 < 5 \land \neg(4 < 5) \land 4 < 5 = \neg(4 < 5)$ which both do not hold. This indicates the paths have diverged and so the algorithm now attaches a new node instead of following the existing path. Then $d$ is increased to $2$. Since there are no $pathModels$ in this new branch, a new node is added. The end result is shown at the bottom of \cref{fig:building-the-tree}.

\begin{figure}
    \centering
    \begin{minipage}[t][6.75cm]{0.49\textwidth}
        \centering
        \begin{tikzpicture}[main/.style = {draw, circle, minimum size=0.5cm},scale=0.8,xscale=0.65]
		\begin{scope}
            \node[main, inner sep=2pt] (1) at (0,0) {$x_0$};
            \node[main, inner sep=2pt] (2) at (3,0) {$x_1$};
            \node[main, inner sep=2pt] (3) at (6,0) {$x_2$};
            \node[main, inner sep=2pt] (4) at (9,0) {$\phantom{x_3}$};
            \draw[->] (1) -- node[midway, above] {$4$} (2);
            \draw[->] (2) -- node[midway, above] {$4$} (3);
            \draw[->] (3) -- node[midway, above] {$4$} (4);
            \node[] (A0) at (4.5, 0.75) {$x_0 < 5 \land x_1 < 5 \land x_2 < 5$};
            \node[anchor=west] (A0) at (9.5, 0) {$M_1$};
        \end{scope}
        \begin{scope}[shift={(0, -1.5)}]
            \node[main, inner sep=2pt] (1) at (0,0) {$x_0$};
            \node[main, inner sep=2pt] (2) at (3,0) {$x_1$};
            \node[main, inner sep=2pt] (3) at (6,0) {$x_2$};
            \node[main, inner sep=2pt] (4) at (9,0) {$\phantom{x_3}$};
            \draw[->] (1) -- node[midway, above] {$3$} (2);
            \draw[->] (2) -- node[midway, above] {$5$} (3);
            \draw[->] (3) -- node[midway, above] {$4$} (4);
            \node[] (A0) at (4.5, 0.75) {$x_0 < 5 \land x_1 \geq 5 \land x_2 < 5$};
            \node[anchor=west] (A0) at (9.5, 0) {$M_2$};
        \end{scope}
        \begin{scope}[shift={(0, -3.5)}]
            \node[main, inner sep=2pt] (1) at (0,0) {$x_0$};
            \node[main, inner sep=2pt] (2) at (3,0) {$x_1$};
			\node[main, inner sep=2pt] (3a) at (6,0.5) {$x_2$};
            \node[main,fill=blue!20, inner sep=2pt] (3b) at (6,-0.5) {$x_2$};
			\node[main, inner sep=2pt] (4a) at (9,0.5) {$\phantom{x_3}$};
            \node[main,fill=blue!20, inner sep=2pt] (4b) at (9,-0.5) {$\phantom{x_3}$};
            \draw[->,line width=0.25mm] (1) -- node[midway, above] {$4$} (2);
            \draw[->] (2) to [out=5,in=185] node[midway, above, yshift=3pt, xshift=0pt] {$4$} (3a);
            \draw[->,blue!50,line width=0.25mm] (2) to [out=-5,in=175] node[midway, below, yshift=-3pt, xshift=0pt] {$5$} (3b);
            \draw[->] (3a) -- node[midway, above] {$4$} (4a);
            \draw[->,blue!50,line width=0.25mm] (3b) -- node[midway, below] {$4$} (4b);
           	\node[anchor=west] (A0) at (9.5, 0.5) {$M_1$};
            \node[anchor=west] (A0) at (9.5, -0.5) {$\textcolor{blue!50}{M_2}$};
			\node[] (A1) at (0, -1) {
				$\begin{bmatrix}
				M_1 \\
				\textcolor{blue!50}{M_2}
				\end{bmatrix}$
			};
			\node[] (A1) at (3, -1) {
				$\begin{bmatrix}
				M_1 \\
				\textcolor{blue!50}{M_2}
				\end{bmatrix}$
			};
            \node[] (A1) at (6, -1.2) {
				$\begin{bmatrix}
				\textcolor{blue!50}{M_2}
				\end{bmatrix}$
			};
			\node[] (A1) at (6, 1.2) {$[M_1]$};
        \end{scope}
    \end{tikzpicture}
	\caption{Example showing how equivalent parts of two models are joined to create a multiverse tree. Changes to the tree are in blue, and the $pathModels$ are shown in square brackets.}
    \label{fig:building-the-tree}    \end{minipage}
    \hfill
    \begin{minipage}[t][6.75cm]{0.49\textwidth}
        \centering
        \begin{tikzpicture}[main/.style = {draw, circle, minimum size=0.5cm}, scale=0.85, xscale=0.65]
    	\begin{scope}
            \node[main, inner sep=2pt] (1) at (0,0) {$x_0$};
            \node[main, inner sep=2pt] (2) at (3,0) {$x_1$};
            \node[main, inner sep=2pt] (3) at (6,0) {$x_2$};
            \node[main, inner sep=2pt] (4) at (9,0) {$\phantom{x_3}$};
            \draw[->] (1) -- node[midway, above] {$2$} (2);
            \draw[->] (2) -- node[midway, above] {$5$} (3);
            \draw[->] (3) -- node[midway, above] {$5$} (4);
            \node[] (A0) at (4.5, 0.75) {$x_0 < 5 \land x_1 \geq 5 \land x_2 \geq 5$};
            \node[anchor=west] (A0) at (9.5, 0) {$M_3$};
        \end{scope}
        \begin{scope}[shift={(0, -2.5)}]
            \node[main, inner sep=2pt] (1) at (0,0) {$x_0$};
            \node[main, inner sep=2pt] (2) at (3,0) {$x_1$};
			\node[main, inner sep=2pt] (3a) at (6,1) {$x_2$};
            \node[main, inner sep=2pt] (3b) at (6,-1) {$x_2$};
			\node[main, inner sep=2pt] (4a) at (9,1) {$\phantom{x_3}$};
            \node[main, inner sep=2pt] (4b) at (9,-0.5) {$\phantom{x_3}$};
            \node[main,fill=blue!20, inner sep=2pt] (4c) at (9,-1.5) {$\phantom{x_3}$};
            \draw[->,line width=0.25mm] (1) -- node[midway, above] {$4$} (2);
            \draw[->] (2) to [out=5,in=190] node[midway, above, yshift=4pt, xshift=-3pt] {$4$} (3a);
            \draw[->,line width=0.25mm] (2) to [out=-5,in=170] node[midway, below, yshift=-4pt, xshift=-3pt] {$5$} (3b);
            \draw[->] (3a) -- node[midway, above] {$4$} (4a);
            \draw[->] (3b) to [out=5,in=185] node[midway, above, yshift=3pt, xshift=0pt] {$4$} (4b);
            \draw[->,blue!50,line width=0.25mm] (3b) to [out=-5,in=175] node[midway, below, yshift=-3pt, xshift=0pt] {$5$} (4c);
            \node[anchor=west] (A0) at (9.5, 1) {$M_1$};
            \node[anchor=west] (A0) at (9.5, -0.5) {$M_2$};
            \node[anchor=west] (A0) at (9.5, -1.5) {\textcolor{blue!50}{$M_3$}};
            \node[] (A1) at (0, -1.25) {
				$\begin{bmatrix}
				M_1 \\
				M_2 \\
				\textcolor{blue!50}{M_3}
				\end{bmatrix}$
			};
            \node[] (A1) at (3, -1.25) {
				$\begin{bmatrix}
				M_1 \\
				M_2 \\
				\textcolor{blue!50}{M_3}
				\end{bmatrix}$
			};
            \node[] (A1) at (6, -2) {
				$\begin{bmatrix}
				M_2 \\
				\textcolor{blue!50}{M_3}
				\end{bmatrix}$
			};
			\node[] (A1) at (6, 1.7) {$[M_1]$};
		\end{scope}
    \end{tikzpicture}
	\caption{The tree after adding the third model. As shown, a new node was added for $x_2 = 5$ in the $M_2$ branch because $M_3$ was equivalent to $M_2$ with respect to $x_1$.}
    \label{fig:building-the-tree-2}
    \end{minipage}
\end{figure}

To further illustrate how the $pathModels$ are used, suppose the concolic execution engine finds a third model $\{x_0=2,x_1=5,x_2=5\}$ with path condition $\pi = x_0 < 5 \land \neg(x_1 < 5) \land \neg(x_2 < 5)$. This model is shown at the top of \cref{fig:building-the-tree-2}. When this model is used to extend the tree, the algorithm will see that the edge connecting the first and second node is equivalent with respect to $x_0$. When coming to the second node however, there are two branches, and two path models to choose from ($M_1$ and $M_2$ from the end result in \cref{fig:building-the-tree}). Only one of these path models is equivalent with respect to $x_1$, namely $M_2$. As a result, the algorithm will follow the edge $mock \; 5$ and attach a new node for $x_2 = 5$ since there does not exist an equivalent path in the $pathModels$ for this node ($[M_2]$). The resulting tree is shown at the bottom of \cref{fig:building-the-tree-2}.

The algorithm keeps repeating this process until all possible execution paths have been explored and the multiverse tree has been constructed. Afterwards, the associated models stored in each of the nodes are removed and the tree is returned.

\subsubsection{Infinite Loops}
Considering many microcontroller programs use an infinite loop, our algorithm only generates trees up to a limited depth. This depth can be configured by limiting the number of instructions in a concolic iteration or by limiting the number of symbolic variables in an iteration. In the pseudocode this can be achieved by limiting the number of $\Rightarrow_{cs}$ transitions, for brevity we only show the former. Aside from limiting the depth, it is also possible to limit the number of concolic iterations, this is handled using the $belowLimit()$ function in the pseudocode. All these options are configurable in the user interface, allowing the developer to analyse the program for a longer period of time if needed.

\subsection{Debugger Guarantees}
The formalisation of our technique follows previous works on multiverse debugging, specifically Voyager~\cite{torres19} and MIO~\cite{mio}, where the debugger semantics are defined in terms of the underlying language semantics.
This allows \emph{debugger correctness} with respect to the underlying language to be defined as the combination of soundness and completeness, as proposed by Lauwaerts et al.~\cite{mio}.
Soundness means that the debugger should only be able to explore states possible in the underlying language semantics.
Completeness states that every state possible in the underlying language semantics should be explorable in the debugger.

Soundness follows because the debugger's $step$ operation is defined in terms of the underlying language semantics, and the $mock$ operation can only mock values within the domain of the primitives.
Completeness also holds because each underlying language step ($\hookrightarrow_{i}$) has a corresponding debugger step ($\rightarrow$). 
For deterministic instructions, the $step$ message can be used. For nondeterministic primitives, the $mock$ message can be used to make the execution in the debugger follow the exact same path as the underlying language semantics.  
Additionally, even though concolic execution suggests only a subset of possible sensor inputs to reduce the state-space, the debugger remains complete because the user can still manually $mock$ other possible inputs.

A full proof, together with the formal definitions of soundness and completeness, is included in \cref{app:guarantee}. This proof differs from MIO due to our trace-based approach and the lack of reversibility. Due to this lack of reversibility we instead rely restarting the execution entirely to slide to a different universe. 

\section{Remote Concolic Multiverse Debugging in Practice}\label{sec:evaluation}

To validate how well our approach works in practice, we implemented a prototype on top of the WARDuino~\cite{lauwaerts24a} WebAssembly virtual machine.
We selected programs out of the standard example programs provided by Arduino\footnote{\url{https://github.com/arduino/arduino-examples/tree/main/examples}}, that use at least \emph{one analog} or \emph{digital input sensor} and have at least \emph{one if condition}. 
We excluded programs that only used output or used sensors that are currently not supported by the WARDuino virtual machine. 
All of the programs we selected are written with an infinite loop which is very common in microcontroller programs. 
We also investigated the state-space reduction on two more elaborate applications, a gesture detector for controlling a robot from the electronics forum electronicsforu.com  and a breakout game that we wrote ourselves. 
We evaluated for all programs how well our analysis was able to reduce the number of options presented to the user in comparison to the state-space shown by traditional multiverse debuggers.
Additionally we discuss the benefits of remote concolic multiverse debugging for each of these programs.

\subsection{State-space Reduction}
\label{appendix:example-programs}

\begin{table*}
\centering
\begin{tabular}{l S[table-format=1.1] r r | r r r r}
Program & {States} & Paths & Max opts & Loop iterations & Time (s) \\
\hline
arduino-crystal-ball & 32 & 11 & 8 & 2 & 0.283 \\                             
arduino-knock & 4096 & 2 & 2 & 1 & 0.038 \\                                   
arduino-touch-sensor-lamp & 4096 & 2 & 2 & 1 & 0.003 \\                       
arduino-switch & 4096 & 4 & 4 & 1 & 0.131 \\                                  
arduino-keyboard & 4096 & 5 & 5 & 1 & 0.116 \\                                
arduino-love-o-meter & 4096 & 4 & 4 & 1 & 1.172 \\							  
arduino-while-no-calibrate & 4096 & 3 & 3 & 1 & 0.063 \\                      
arduino-while* & 33554432 & 76 & 65 & 1 & 219.613 \\                          
arduino-knock-lock & 33554432 & 13 & 3 & 2 & 0.074 \\                         
arduino-zoetrope & 268435456 & 16 & 2 & 2 & 0.080 \\                          
\hline
gesture-robot & 1329227995784915872903807060280344576 & 31 & 2 & 1 & 0.106 \\ 
breakout & 399233084044151956813969209973956763668805430970279476309364211949286649982748318240784572758907890368640881602991649656625549381476314280797456433584484495323604418213761551729681918627066399886090164997382778091838902730878898388605362171524454751878181477200406028946841317268160892562856076099796409527278605494960759076784916422446271871145391866999805548212763222246676828598760892071038179814502510727747623556368593631835054902213439370817372357339955937546856274232357677693777689545591891728973416173348886163644903350963610044575631604269786233734989471537832854231209470034546477344877119101940726018732915614418402948719509270679375079654525623106166646842000235771760798823878229459247865679922617648576081806794776733695402348032190905056511798831636005476500612696383661139062605947521920016712747015023300934148902057804293592161253081272918237389877376177664087548874960911121381970294611560924151873555719584277856256 & 3 & 3 & 259 & 5.032 \\
\end{tabular}
\caption{Analysis of the state-space reduction algorithm showing the estimated number of states, the total number of execution paths found and the maximum number of options presented at one choice point in the debugger. For each of the programs we also provide the number of loop iterations analysed and the execution time. Full source links can be found in~\cref{app:sources}.}
\label{tbl:state-space-reduction}
\end{table*}

The data presented in \Cref{tbl:state-space-reduction} shows an overview of the state-space reduction for all programs that we tested. 
We under estimate the state space of a program by taking the path in the program with the most choice points and multiplying the number of options along this path (\textbf{States}).
This number estimates the number of states that would be presented to the user in a traditional multiverse debugger. 
We show the amount of unique paths (\textbf{Paths}) found by our analysis and the maximum number of options (\textbf{Max opts}) presented to the user at any given point in time. 
For programs that use state from previous iterations we executed the loop twice as this is the minimal iteration count needed to see an effect. For programs that do not use state from previous iterations we limited the analysis to just one loop iteration (\textbf{Loop iterations}). 
For each of these configurations we also give the time (\textbf{Time (s)}) it took to analyse the program.
As the data shows, our technique significantly reduces the state-space for each of the programs we tested\footnote{The arduino-while program has two variants, we discuss the reason behind this in more detail later.}.
In the following sections, we give an overview of these programs and show how remote concolic multiverse debugging helps to debug them.

\subsection{Low Complexity}

\begin{figure}
  \centering
  \includegraphics[scale=0.5,valign=c,width=0.39\linewidth]{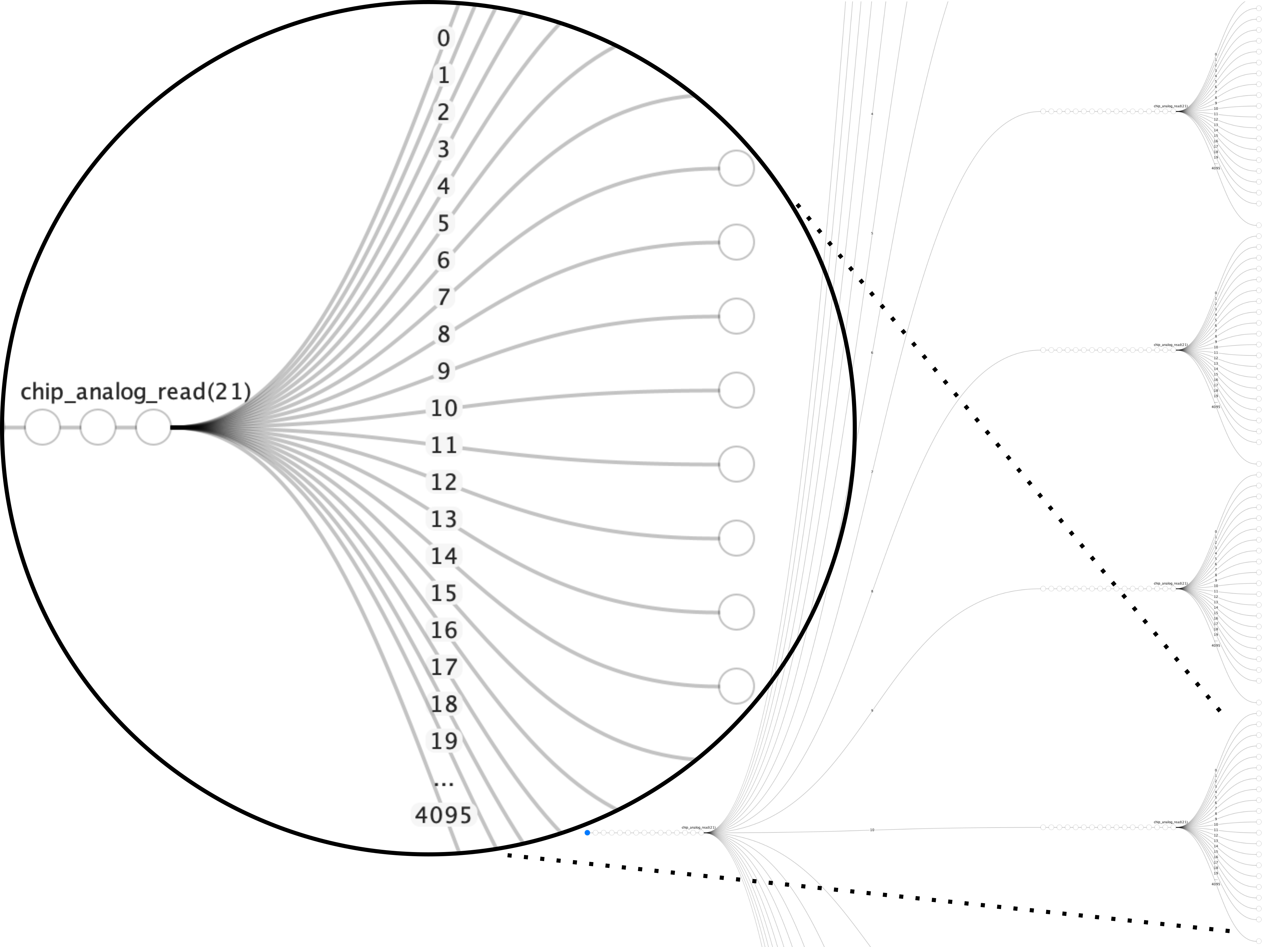}%
  \vrule
  \includegraphics[scale=0.3,valign=c,width=0.6\linewidth]{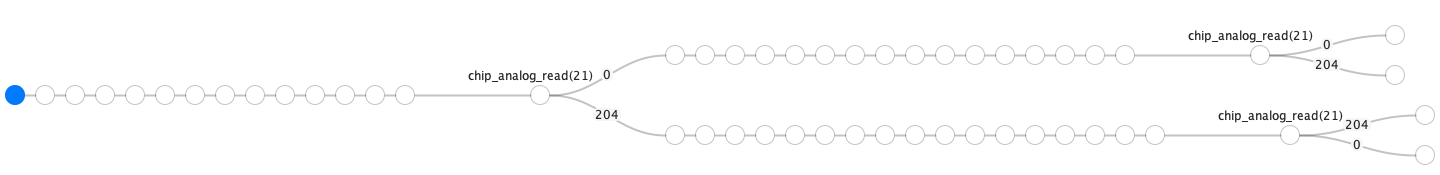}
  \caption{Left, before state-space reduction. Right, after state-space reduction.}
  \label{fig:compare-reduction}
\end{figure}

The \textbf{\href{https://github.com/arduino/arduino-examples/blob/main/examples/06.Sensors/Knock/Knock.ino}{arduino-knock}}, \textbf{\href{https://github.com/arduino/arduino-examples/blob/main/examples/10.StarterKit_BasicKit/p13_TouchSensorLamp/p13_TouchSensorLamp.ino}{arduino-touch-sensor-lamp}}, \textbf{\href{https://github.com/arduino/arduino-examples/blob/main/examples/05.Control/switchCase/switchCase.ino}{arduino-switch}}, \textbf{\href{https://github.com/arduino/arduino-examples/blob/main/examples/10.StarterKit_BasicKit/p07_Keyboard/p07_Keyboard.ino}{arduino-keyboard}} and \textbf{\href{https://github.com/arduino/arduino-examples/blob/main/examples/10.StarterKit_BasicKit/p03_LoveOMeter/p03_LoveOMeter.ino}{arduino-love-o-meter}} programs only have 2 to 5 paths. 
While the state-space for these programs is still relatively big, a quick inspection suffices to find good nondeterministic values to choose from.
The advantage of using our debugger over existing multiverse debuggers such as MIO~\cite{mio} is that the user can much more easily explore interesting paths. Exploring a new execution path traditionally requires manually adding each edge through individual $\mathit{mock}$ and $\mathit{step}$ operations. Additionally the user also has to decide which values they use when mocking from a large amount of options as visible in the left of \cref{fig:compare-reduction}. Using our approach, all future execution paths can be created automatically with a single button press. 
This results in a tree where each execution path has corresponding inputs to explore it as shown on the right in \cref{fig:compare-reduction}. 
In the \textbf{\href{https://github.com/arduino/arduino-examples/blob/main/examples/06.Sensors/Knock/Knock.ino}{arduino-knock}} program, this reduces the state-space from $4096^2$ paths to just $4$. 
The program uses a simple \texttt{if} and $2$ values are needed to explore both branches.
In two loop iterations this results in $4$ paths.
To explore an execution path, one simply clicks on a node in the path and slides to it, performing all the needed mocking and stepping operations along the way. This also allows users to easily explore sequences of values instead of just one singular value at a time as with mocking.
For example, using the slide operation on the created paths in the \textbf{\href{https://github.com/arduino/arduino-examples/blob/main/examples/10.StarterKit_BasicKit/p07_Keyboard/p07_Keyboard.ino}{arduino-keyboard}} program which implements a basic piano, sequences of notes can easily be explored. Furthermore, our analysis helps to point out strange behaviors in programs. For example, the \textbf{\href{https://github.com/arduino/arduino-examples/blob/main/examples/05.Control/switchCase/switchCase.ino}{arduino-switch}} program reads an analog value and maps it from 0 to 3 (inclusive) for use in a switch. Strangely this program has 5 execution paths instead of 4. This is due to the arduino \texttt{map} function which does not constrain the output when the input value is higher than the maximum. Our analysis points this out to the developer. Finally, our analysis also helps in cases where the value needed to mock is unclear, such as in the \textbf{\href{https://github.com/arduino/arduino-examples/blob/main/examples/10.StarterKit_BasicKit/p03_LoveOMeter/p03_LoveOMeter.ino}{arduino-love-o-meter}} where floating points calculations convert an analog value into celsius before performing branching. In this case our analysis calculates the right input values needed to get into the right branch after the conversion. Calculating the needed values by hand would take a significant effort.


\subsection{Medium Complexity}
The second set of programs have a higher path complexity ranging between 11 to 16 paths.
We found that for some of these programs the state-space reduction becomes significant, as it is no longer obvious which input values are needed for full coverage based on a quick inspection of the source code. 
The main reason is that these programs combine multiple sources of non-determism and keep track of state over multiple loop iterations.
This combination makes it much harder to reason about the possible combinations that can occur during execution. 

The \textbf{\href{https://github.com/arduino/arduino-examples/blob/main/examples/10.StarterKit_BasicKit/p11_CrystalBall/p11_CrystalBall.ino}{arduino-crystal-ball}} has 11 paths because the program implements a magic 8 ball using a random integer from 0 to 8 (exclusive) with a switch case and a button that needs to be pressed and released. This program illustrates how programs using the \texttt{random} function can be more easily debugged because the behaviour of this nondeterministic function can be controlled using the debugger.
The \textbf{\href{https://github.com/arduino/arduino-examples/blob/main/examples/10.StarterKit_BasicKit/p12_KnockLock/p12_KnockLock.ino}{arduino-knock-lock}} program implements a vault door that can be opened by knocking a specific pattern. It has 13 paths in two iterations, this is because it allows users to either not press a button or knock in 3 different ways, this results in 4 paths in the first iteration and $4 + 3 \cdot 3 = 13$ in the second iteration. This program opens the lock if the door is knocked on 3 times, as such a higher iteration count might be interesting in this program. This is configurable in the debugger. Using our analysis, the user can easily generate the needed values to trigger different kinds of knocks and explore the effect of knocking various different sequences without actually needing to manually knock each sequence.
Finally, the \textbf{\href{https://github.com/arduino/arduino-examples/blob/main/examples/10.StarterKit_BasicKit/p10_Zoetrope/p10_Zoetrope.ino}{arduino-zoetrope}} program implements a zoetrope which provides the illusion of moving images. It has 16 paths, this is because it has two buttons that can both be pressed or not pressed in each iteration resulting in 4 paths in the first iteration and 16 in the second. Given this program only branches on digital inputs, it is relatively easy to debug. However, this program also showcases an additional aspect of our debugger. To control the motor speed, the program directly uses an analog value read from a potentiometer. Since this is only a single execution path, only a single value is generated. However, using our manual mocking operation, other additional inputs can easily be explored allowing developers to reproduce bugs where a different speed is required. The same applies to many other programs where an input directly controls the output without branching.

\subsection{High Complexity}
We categorize three of the selected programs as high complexity, either due to the many paths generated or due to the larger code base making it more complex to understand quickly. 
These three programs are the arduino-while example, a gesture controlled robot and the breakout game.


The \href{https://github.com/arduino/arduino-examples/blob/main/examples/05.Control/WhileStatementConditional/WhileStatementConditional.ino}{\textbf{arduino-while}} program first calibrates a sensor at runtime when a button is pressed. 
Subsequently, it reads an analog value and runs it through a formula with hardcoded values. 
It then clamps the result between 0 and 255 using two if statements. 
Interestingly, this program generates an abnormally large number of paths.
The amount of states for this program are so excessive that we had to limit the number of concolic iterations for the analysis to finish in a reasonable time. 
This abnormally large number of paths can be attributed to the calibration step used in this program which first reads an analog value which later serves as a threshold in the program. 
Because this sensor value is used to determine the branching conditions, the concolic analysis can keep finding new paths by simply choosing a different sensor value in the calibration, resulting in a nearly infinite number of paths.

Our work overcomes this issue by letting the user first run the program and calibrate the sensor value and then debug the rest of the application given the chosen calibration. 
This way, once the user chooses the path that will be taken depending on this unknown value, the value is already chosen and only few options remain. 
Our work enables this by employing a dynamic approach that starts the analysis from the current program state in the debugger instead of using a static analysis that does not take the current program state into account.
To show the number of paths after choosing a calibration value, we also evaluated a version of this program without the calibration called arduino-while-no-calibrate which can be found in \cref{app:while-no-calibrate}. This is equivalent to choosing a calibration value in the debugger and then running the analysis.
As shown in \cref{tbl:state-space-reduction}, this drastically reduces the number of options, from a nearly infinite collection of options to just 3 paths once the calibration phase is over. These three paths are a result of the clamping operation, either the value is within the bounds, below the minimum or above the maximum value.

The \href{https://www.electronicsforu.com/electronics-projects/build-clap-gesture-controlled-robot}{\textbf{gesture-robot}} program implements a gesture detector to remotely control a robot.
The remote works using an accelerometer, analog values are read to get the acceleration on the x and y axis. 
Based on these values the remote can make the robot drive forwards, backwards, turn left, right or stay idle.

This program uses a chain of four conditional if else branches with an else at the end. 
Each if condition reads in a new x and y axis acceleration value, when both these values are within the required threshold the robot performs the associated action. Each if statement has 3 options, either both x and y are within bounds, only one is or none are. For two of these the program falls back to the next if statement which then again has 3 options, in the end this results in $1 + 2 \cdot (1 + 2 \cdot (1 + 2 \cdot (1 + 2))) = 31$ possible paths.
Each of these paths is correctly discovered by our concolic analysis.
The main complexity of the gesture-robot application stems from the way the program is written. 
The program reads in a new sensor value whenever a sensor value is needed instead of storing sensor values in local variables. 
As a result it is possible for sensor values to change in between reading values causing strange unpredictable behavior. 
Our debugger makes the excessive number of possibilities clear to the developer and allows them to take action. Additionally, besides helpful to point out these issues, the analysis allows the developer to easily test all the possible gestures without having to manually perform each action in the real world.

The \textbf{breakout game} listed in \cref{app:breakout} implements a simple breakout game where an analog joystick determines the position of the paddle. 
When analyzing this program for 259 loop iterations, the point where the ball will reach the bottom of the screen, the analysis determined that there are only three paths. 
One in which the paddle is on the left of the ball, one where it hits the paddle and one where it is on the right of the paddle. 
The analysis also found that in a large portion of the program, the position of the joystick does not matter, resulting in a single execution path. 
This is the case when the ball is not near the paddle and the paddle can thus not hit the ball. Using the provided paths, the developer can easily test certain game scenarios such as behavior at very high scores which might be very difficult to reach for a developer.
Due to the complexity and size of this program we believe this program showcases that the analysis can be effective over a larger program with a large number of iterations. Additionally, while the snapshot based approach of MIO~\cite{mio} works well for most programs, we noticed a large reduction in overhead when using our trace-based debugger for this particular program. In a test on the breakout program, it took around 1.9s to execute $50 000$ instructions without tracing. When using our trace-based approach the execution time went up to 2s. However, MIO's checkpointing policy (which in this case takes checkpoints every $\pm 2$ instructions) increased the execution time to around 5 minutes for this particular program. This shows that the overhead of our trace-based approach is significantly less than MIO for programs that perform very frequent IO operations. However, while the overhead of forward execution is significantly reduced, the time to step back has increased as it requires re-executing the program. A more in depth analysis is provided in \cref{app:performance}.

\section{Related Work}\label{sec:related}

While combining static analysis with debugging is a very novel space, there are many domains and specific works connected to our efforts.

\subparagraph*{Multiverse Debugging}
Current multiverse debuggers can be categorized along three dimensions. First, the execution model: offline models of the program or online concrete execution. Second, the way state explosion is handled: not at all, collapsing equivalent states, or reducing the number of paths. Finally, the way they manage the execution state: tracking all states, checkpointing with snapshots, or through traces.

The first multiverse debugger was introduced by Torres Lopez et al. \cite{torres19} in the context of parallel actor-based systems.
It was implemented in the prototype debugger Voyager~\cite{gurdeep19a} that operates directly on a language's operational semantics specified in PLT Redex~\cite{felleisen09}.
In terms of the three architectural dimensions, the multiverse debugger by Torres Lopez et al. \cite{torres19} operates over a model of the program execution, keeps track of all states in a program, and does not deal with the state explosion problem.
Subsequent research has extended multiverse debugging in all three of these areas.
Pasquier et al. \cite{pasquier22} were the first to address the state explosion problem, introducing user-defined reduction rules to quotient the state space in order to accelerate multiverse-wide breakpoint lookup. In subsequent work \cite{pasquier23}, they also introduce temporal breakpoints based on linear temporal logic.
In contrast, our work does not group together similar nodes, it instead reduces the number of paths in the multiverse by leveraging concolic execution. Both of these approaches are complementary and combining them could be highly beneficial to further improve the practical usability of multiverse debuggers.

The works by Pasquier et al. \cite{pasquier22, pasquier23} still operate over a model of the program, and keep track of all states in the quotient state space.
In contrast our work uses an online concrete execution model rather than a model of the program execution.
The work by Lauwaerts et. al.~\cite{mio} introduced the first multiverse debugger operating on concrete program executions, the MIO debugger. MIO does not track all program states but instead uses a checkpointing system.
The work addresses the challenge of performing Input-Output (I/O) operations with external effects during concrete multiverse debugging.
Their solution introduces deterministically reversible primitives with compensating actions, allowing the debugger to step back without creating inconsistent states.
However, the work does not deal with the state explosion problem.
In contrast to MIO which focusses on IO consistency, our work focusses on the orthogonal problem of state-space pruning.
We propose a novel concolic multiverse debugger that intelligently reduces the state-space of nondeterministic programs while maintaining full code coverage. Unlike MIO our work uses a trace-based approach to multiverse debugging instead of using checkpointing which trades reversibility for performance.
In particular, our trace-based approach significantly reduces the performance overhead of forward execution in comparison to MIO for programs that use a high frequency of IO operations. While the forward execution performance is improved, exploring alternative paths is slower since this approach requires restarting the execution instead of stepping back.

\subparagraph*{Exploring Execution Trees}
Exploration of program execution trees is not exclusive to multiverse debuggers.
Many analysis tools likewise explore execution trees of programs, such as software \emph{model checkers}~\cite{godefroid97,jhala09}, \emph{symbolic execution}~\cite{king76,cadar11,baldoni18}, and \emph{concolic execution}~\cite{godefroid05,sen06,marques22}.
While these methods excel at identifying program defects automatically, they require an explicit problem specification or program description, typically expressed as a formal model.
In contrast debuggers assist developers in finding errors that cannot be precisely formulated, or where the cause of the bug is unknown.
We believe that static analysis methods can substantially improve debugging tools by supplying developers with additional information.
This is the approach that we took in this paper, where we have focused on how \emph{concolic execution} can help reduce the state explosion problem~\cite{valmari98,kurshan98,kahlon09} by reducing the number of execution paths that the programmer needs to consider in the debugger.

\subparagraph*{Remote Debugging on Microcontrollers}
The prototype of our concolic multiverse debugger is implemented as a remote debugger on the WARDuino microcontroller virtual machine~\cite{lauwaerts24a}.
Remote debugging is the most widely used approach on embedded systems~\cite{potsch17,skvar-c24,soderby24} since it can mitigate some limitations of microcontrollers.
In remote debugging~\cite{rosenberg96}, a debugger frontend is connected to a remote debugger backend, or stub, running the program being debugged.
Instead of a stub, embedded debuggers also commonly rely on dedicated hardware such as the JTAG~\cite{shortm} hardware debugger~\cite{hogl06}.
However, remote debugging can exacerbate the probe effect~\cite{gait86}, and can be very slow since the debugger runs on the microcontroller, combined with constant communication overhead.
To address these limitations, a technique called out-of-place debugging, has been proposed~\cite{marra18,lauwaerts22}.
Using this technique, part of the debugger can run on a more powerful machine, reducing debugging interference and improving performance.
Our debugger, is currently not a full fledged out-of-place debugger, but inspired by these techniques we offload the static analysis to a more powerful machine. 
This approach is already sufficiently fast, but a speed-up can likely be achieved by adopting more ideas of out-of-place debugging to increase the amount of code that is offloaded. 

\subparagraph*{Multiverse Analysis}
Analyzing the multiverse of possibilities is more widely known as multiverse analysis, for instance within statistical analysis~\cite{steegen16}.
Within software development, several frameworks for exploratory programming~\cite{kery17a} allow developers to interact with the multiverse of source code versions~\cite{steinert12}.
Programmers actively explore the behavior of a program by experimenting with different code often through dedicated \emph{explore-first IDEs} with advanced version control~\cite{steinert12,kery17}.
While explore-first editors consider variations in the program code itself, multiverse debuggers focus on the various possible executions for a a single instance of a program.
The combination of these two techniques could result in a powerful development environment.

\subparagraph*{Formalizing Debuggers}
Only a limited collection of works consider formalizing the operation of debuggers, however in recent years the approach first used by \cite{bernstein95a} has become the most widely used~\cite{ferrari01, torres17, torres19, lauwaerts24a, holter24}.
It defines the operation of a debugger in terms of an operational semantics that encapsulates an underlying language semantics.
We have followed the same recipe here since it allows the formalism to explicitly show how the concrete WebAssembly execution relates to the concolic analysis in the debugger.

Recently Holter et al. \cite{holter24} introduced a novel abstract debugger that leverages static analysis to let developers explore abstract program states instead of concrete ones.
The work defines operational semantics for both an abstract and a concrete debugger.
The abstract debugger is proved sound with respect to the concrete one; in other words, every concrete debug session is guaranteed to correspond to an abstract session.
The converse does not hold because the over-approximations in the static analysis, so the abstract semantics may admit sessions that cannot occur concrete world.
In contrast our work leverages concolic execution, which eliminates the need for over-approximation. 
While the overall debugger is complete because of the existence of a manual $mock$ operation, we lose completeness of the concolically explored paths, as we must set bounds to analyse infinite programs.

\section{Conclusion}\label{sec:conclusion}
In this article, we introduced remote concolic multiverse debugging, a novel hybrid debugging approach which allows pruning the state-space of nondeterministic programs in multiverse debuggers.
We have demonstrated the power of this technique by applying it to a diverse set of resource-constrained microcontroller programs, highlighting its ability to improve the debugging process across a wide range of applications. 
To scale the approach to concrete executions on microcontrollers we developed a remote debugger (on top of the WARDuino VM) that uses a trace-based approach to multiverse debugging rather than the conventional snapshot-based approach taken by previous works.
We have formalized this approach and provide a proof of the soundness and completeness.
 At the heart of our debugging visualisation lies a novel algorithm designed to construct a multiverse of program executions from the results of concolic execution. 
This algorithm effectively enables a systematic exploration of all execution paths with our debugger.

Our work demonstrates that combining online static analysis with live debugging is the key to managing the large state spaces of complex programs. 
We believe this structured integration points toward a promising new direction for more effective debugging tools. 
As the first working example of this approach for microcontrollers, we hope our foundational research will inspire the exploration of a wider design space for debuggers, encouraging the community to pair other static analysis methods with advanced debugging techniques and ultimately advance the way we debug complex software.

\bibliographystyle{plainurl}
\bibliography{bibliography}

@article{mcdowell89,
	title        = {Debugging Concurrent Programs},
	author       = {McDowell, Charles E. and Helmbold, David P.},
	year         = 1989,
	month        = dec,
	journal      = {ACM Comput. Surv.},
	volume       = 21,
	number       = 4,
	pages        = {593--622},
	doi          = {10.1145/76894.76897},
	issn         = {0360-0300}
}

@phdthesis{gurdeep22,
	title        = {Taming Nondeterminism : Programming Language Abstractions and Tools for Dealing with Nondeterministic Programs},
	author       = {Gurdeep Singh, Robbert},
	year         = 2022,
	school       = {Ghent University}
}

@inproceedings{torres19,
	title        = {Multiverse {{Debugging}}: {{Non-Deterministic Debugging}} for {{ Non-Deterministic Programs}} ({{Brave New Idea Paper}})},
	author       = {Torres Lopez, Carmen and Gurdeep Singh, Robbert and Marr, Stefan and Gonzalez Boix, Elisa and Scholliers, Christophe},
	year         = 2019,
	booktitle    = {{{DROPS-IDN}}/v2/Document/10.4230/{{LIPIcs}}.{{ECOOP}}.2019.27},
	publisher    = {Schloss Dagstuhl -- Leibniz-Zentrum f{\"u}r Informatik},
	doi          = {10.4230/LIPIcs.ECOOP.2019.27}
}

@inproceedings{pasquier22,
	title        = {Practical Multiverse Debugging through User-Defined Reductions: Application to {{UML}} Models},
	author       = {Pasquier, Matthias and Teodorov, Ciprian and Jouault, Fr{\'e}d{\'e}ric and Brun, Matthias and Roux, Luka Le and Lagadec, Lo{\"i}c},
	year         = 2022,
	month        = oct,
	booktitle    = {Proceedings of the 25th {{International Conference}} on {{Model Driven Engineering Languages}} and {{Systems}}},
	publisher    = {Association for Computing Machinery},
	address      = {New York, NY, USA},
	series       = {{{MODELS}} '22},
	pages        = {87--97},
	doi          = {10.1145/3550355.3552447},
	isbn         = {978-1-4503-9466-6}
}

@inproceedings{pasquier23,
	title        = {Temporal {{Breakpoints}} for {{Multiverse Debugging}}},
	author       = {Pasquier, Matthias and Teodorov, Ciprian and Jouault, Fr{\'e}d{\'e}ric and Brun, Matthias and Le Roux, Luka and Lagadec, Lo{\"i}c},
	year         = 2023,
	month        = oct,
	booktitle    = {Proceedings of the 16th {{ACM SIGPLAN International Conference} } on {{Software Language Engineering}}},
	publisher    = {Association for Computing Machinery},
	address      = {New York, NY, USA},
	series       = {{{SLE}} 2023},
	pages        = {125--137},
	doi          = {10.1145/3623476.3623526},
	isbn         = 9798400703966
}

@inproceedings{pasquier23a,
	title        = {Debugging {{Paxos}} in the {{UML Multiverse}}},
	author       = {Pasquier, Matthias and Teodorov, Ciprian and Jouault, Fr{\'e}d{\'e}ric and Brun, Matthias and Lagadec, Lo{\"i}c},
	year         = 2023,
	month        = oct,
	booktitle    = {2023 {{ACM}}/{{IEEE International Conference}} on {{Model Driven Engineering Languages}} and {{Systems Companion}} ({{ MODELS-C}})},
	pages        = {811--820},
	doi          = {10.1109/MODELS-C59198.2023.00130}
}

@article{gait86,
	title        = {A Probe Effect in Concurrent Programs},
	author       = {Gait, Jason},
	year         = 1986,
	journal      = {Software: Practice and Experience},
	volume       = 16,
	number       = 3,
	pages        = {225--233},
	doi          = {10.1002/spe.4380160304},
	issn         = {1097-024X}
}

@book{zeller05,
  title        = {Why {{Programs Fail}}: {{A Guide}} to {{Systematic
		              Debugging}}},
  shorttitle   = {Why {{Programs Fail}}},
  author       = {Zeller, Andreas},
  year         = {2005},
  month        = sep,
  publisher    = {Morgan Kaufmann Publishers Inc.},
  address      = {San Francisco, CA, USA},
  isbn         = {978-1-55860-866-5}
}

@inproceedings{khan11,
  title		= {Limitations of Simulation Tools for Large-Scale Wireless
		  Sensor Networks},
  booktitle	= {2011 {{IEEE}} Workshops of International Conference on
		  Advanced Information Networking and Applications},
  author	= {Khan, Muhammad Zahid and Askwith, Bob and Bouhafs, Faycal
		  and Asim, Muhammad},
  year		= {2011},
  pages		= {820--825},
  publisher	= {IEEE},
  address	= {New York, NY, USA},
  doi		= {10.1109/WAINA.2011.59}
}

@article{roska90,
  title		= {Limitations and Complexity of Digital Hardware Simulators
		  Used for Large-Scale Analogue Circuit and System Dynamics},
  author	= {Roska, Tam{\'a}s},
  year		= {1990},
  journal	= {International Journal of Circuit Theory and Applications},
  volume	= {18},
  number	= {1},
  pages		= {11--21},
  issn		= {1097-007X},
  doi		= {10.1002/cta.4490180104},
  urldate	= {2024-08-31},
  copyright	= {Copyright {\copyright} 1990 John Wiley \& Sons, Ltd.},
  langid	= {english}
}

@article{perscheid17,
  title		= {Studying the Advancement in Debugging Practice of
		  Professional Software Developers},
  author	= {Perscheid, Michael and Siegmund, Benjamin and Taeumel,
		  Marcel and Hirschfeld, Robert},
  year		= {2017},
  month		= mar,
  journal	= {Software Quality Journal},
  volume	= {25},
  number	= {1},
  pages		= {83--110},
  issn		= {1573-1367},
  doi		= {10.1007/s11219-015-9294-2},
  urldate	= {2025-01-13},
  langid	= {english}
}

@inproceedings{makhshari21,
  title		= {{{IoT Bugs}} and {{Development Challenges}}},
  booktitle	= {2021 {{IEEE}}/{{ACM}} 43rd {{International Conference}} on
		  {{Software Engineering}} ({{ICSE}})},
  author	= {Makhshari, Amir and Mesbah, Ali},
  year		= {2021},
  month		= may,
  pages		= {460--472},
  issn		= {1558-1225},
  doi		= {10.1109/ICSE43902.2021.00051},
  urldate	= {2025-01-13}
}

@article{lauwaerts24a,
	title        = {{{WARDuino}}: {{An}} Embedded {{WebAssembly}} Virtual Machine},
	author       = {Lauwaerts, Tom and Singh, Robbert Gurdeep and Scholliers, Christophe},
	year         = 2024,
	month        = feb,
	journal      = {Journal of Computer Languages},
	pages        = 101268,
	doi          = {10.1016/j.cola.2024.101268},
	issn         = {2590-1184}
}

@inproceedings{haas17,
	title        = {Bringing the Web up to Speed with {{WebAssembly}}},
	author       = {Haas, Andreas and Rossberg, Andreas and Schuff, Derek L. and Titzer, Ben L. and Holman, Michael and Gohman, Dan and Wagner, Luke and Zakai, Alon and Bastien, {\relax JF}},
	year         = 2017,
	month        = jun,
	booktitle    = {Proceedings of the 38th {{ACM SIGPLAN Conference}} on {{ Programming Language Design}} and {{Implementation}}},
	publisher    = {Association for Computing Machinery},
	address      = {New York, NY, USA},
	series       = {{{PLDI}} 2017},
	pages        = {185--200},
	doi          = {10.1145/3062341.3062363},
	isbn         = {978-1-4503-4988-8}
}

@article{gurdeep19a,
	title        = {Multiverse {{Debugging}}: {{Non-Deterministic Debugging}} for {{ Non-Deterministic Programs}} ({{Artifact}})},
	author       = {Gurdeep Singh, Robbert and Torres Lopez, Carmen and Marr, Stefan and Gonzalez Boix, Elisa and Scholliers, Christophe},
	year         = 2019,
	journal      = {Dagstuhl Artifacts Series},
	publisher    = {Schloss Dagstuhl -- Leibniz-Zentrum f{\"u}r Informatik},
	volume       = 5,
	number       = 2,
	pages        = {4:1--4:3},
	doi          = {10.4230/DARTS.5.2.4},
	issn         = {2509-8195}
}

@book{felleisen09,
	title        = {Semantics Engineering with {{PLT}} Redex},
	author       = {Felleisen, Matthias and Findler, Robert Bruce and Flatt, Matthew},
	year         = 2009,
	publisher    = {Mit Press}
}

@article{steegen16,
	title        = {Increasing {{Transparency Through}} a {{Multiverse Analysis}}},
	author       = {Steegen, Sara and Tuerlinckx, Francis and Gelman, Andrew and Vanpaemel, Wolf},
	year         = 2016,
	month        = sep,
	journal      = {Perspectives on Psychological Science},
	publisher    = {SAGE Publications Inc},
	volume       = 11,
	number       = 5,
	pages        = {702--712},
	doi          = {10.1177/1745691616658637},
	issn         = {1745-6916}
}

@inproceedings{kery17a,
	title        = {Exploring Exploratory Programming},
	author       = {Kery, Mary Beth and Myers, Brad A.},
	year         = 2017,
	month        = oct,
	booktitle    = {2017 {{IEEE Symposium}} on {{Visual Languages}} and {{ Human-Centric Computing}} ({{VL}}/{{HCC}})},
	pages        = {25--29},
	doi          = {10.1109/VLHCC.2017.8103446},
	issn         = {1943-6106}
}

@article{steinert12,
	title        = {{{CoExist}}: Overcoming Aversion to Change},
	author       = {Steinert, Bastian and Cassou, Damien and Hirschfeld, Robert},
	year         = 2012,
	month        = oct,
	journal      = {SIGPLAN Not.},
	volume       = 48,
	number       = 2,
	pages        = {107--118},
	doi          = {10.1145/2480360.2384591},
	issn         = {0362-1340}
}

@inproceedings{kery17,
	title        = {Variolite: {{Supporting Exploratory Programming}} by {{Data Scientists}}},
	author       = {Kery, Mary Beth and Horvath, Amber and Myers, Brad},
	year         = 2017,
	month        = may,
	booktitle    = {Proceedings of the 2017 {{CHI Conference}} on {{Human Factors}} in {{Computing Systems}}},
	publisher    = {Association for Computing Machinery},
	address      = {New York, NY, USA},
	series       = {{{CHI}} '17},
	pages        = {1265--1276},
	doi          = {10.1145/3025453.3025626},
	isbn         = {978-1-4503-4655-9}
}

@inproceedings{godefroid97,
	title        = {Model Checking for Programming Languages Using {{VeriSoft}}},
	author       = {Godefroid, Patrice},
	year         = 1997,
	month        = jan,
	booktitle    = {Proceedings of the 24th {{ACM SIGPLAN-SIGACT}} Symposium on {{ Principles}} of Programming Languages},
	publisher    = {Association for Computing Machinery},
	address      = {New York, NY, USA},
	series       = {{{POPL}} '97},
	pages        = {174--186},
	doi          = {10.1145/263699.263717},
	isbn         = {978-0-89791-853-4}
}

@article{jhala09,
	title        = {Software Model Checking},
	author       = {Jhala, Ranjit and Majumdar, Rupak},
	year         = 2009,
	month        = oct,
	journal      = {ACM Comput. Surv.},
	volume       = 41,
	number       = 4,
	pages        = {21:1--21:54},
	doi          = {10.1145/1592434.1592438},
	issn         = {0360-0300}
}

@article{king76,
	title        = {Symbolic Execution and Program Testing},
	author       = {King, James C.},
	year         = 1976,
	month        = jul,
	journal      = {Communications of the ACM},
	volume       = 19,
	number       = 7,
	pages        = {385--394},
	doi          = {10.1145/360248.360252},
	issn         = {0001-0782}
}

@inproceedings{cadar11,
	title        = {Symbolic Execution for Software Testing in Practice: Preliminary Assessment},
	author       = {Cadar, Cristian and Godefroid, Patrice and Khurshid, Sarfraz and P {\u a}s{\u a}reanu, Corina S. and Sen, Koushik and Tillmann, Nikolai and Visser, Willem},
	year         = 2011,
	month        = may,
	booktitle    = {Proceedings of the 33rd {{International Conference}} on {{ Software Engineering}}},
	publisher    = {Association for Computing Machinery},
	address      = {New York, NY, USA},
	series       = {{{ICSE}} '11},
	pages        = {1066--1071},
	doi          = {10.1145/1985793.1985995},
	isbn         = {978-1-4503-0445-0}
}

@article{baldoni18,
	title        = {A {{Survey}} of {{Symbolic Execution Techniques}}},
	author       = {Baldoni, Roberto and Coppa, Emilio and D'elia, Daniele Cono and Demetrescu, Camil and Finocchi, Irene},
	year         = 2018,
	month        = may,
	journal      = {ACM Comput. Surv.},
	volume       = 51,
	number       = 3,
	pages        = {50:1--50:39},
	doi          = {10.1145/3182657},
	issn         = {0360-0300}
}

@article{godefroid05,
	title        = {{{DART}}: Directed Automated Random Testing},
	author       = {Godefroid, Patrice and Klarlund, Nils and Sen, Koushik},
	year         = 2005,
	month        = jun,
	journal      = {SIGPLAN Not.},
	volume       = 40,
	number       = 6,
	pages        = {213--223},
	doi          = {10.1145/1064978.1065036},
	issn         = {0362-1340}
}

@inproceedings{sen06,
	title        = {Automated {{Systematic Testing}} of {{Open Distributed Programs}}},
	author       = {Sen, Koushik and Agha, Gul},
	year         = 2006,
	booktitle    = {Fundamental {{Approaches}} to {{Software Engineering}}},
	publisher    = {Springer},
	address      = {Berlin, Heidelberg},
	pages        = {339--356},
	doi          = {10.1007/11693017_25},
	isbn         = {978-3-540-33094-3},
	editor       = {Baresi, Luciano and Heckel, Reiko}
}

@inproceedings{marques22,
	title        = {Concolic {{Execution}} for {{WebAssembly}}},
	author       = {Marques, Filipe and Fragoso Santos, Jos{\'e} and Santos, Nuno and Ad{\~a}o, Pedro},
	year         = 2022,
	booktitle    = {36th {{European Conference}} on {{Object-Oriented Programming}} ({{ECOOP}} 2022)},
	publisher    = {Schloss Dagstuhl -- Leibniz-Zentrum f{\"u}r Informatik},
	address      = {Dagstuhl, Germany},
	series       = {Leibniz {{International Proceedings}} in {{Informatics}} ({{LIPIcs }})},
	volume       = 222,
	pages        = {11:1--11:29},
	doi          = {10.4230/LIPIcs.ECOOP.2022.11},
	isbn         = {978-3-95977-225-9},
	issn         = {1868-8969},
	editor       = {Ali, Karim and Vitek, Jan}
}

@incollection{valmari98,
	title        = {The State Explosion Problem},
	author       = {Valmari, Antti},
	year         = 1998,
	booktitle    = {Lectures on {{Petri Nets I}}: {{Basic Models}}: {{Advances}} in {{Petri Nets}}},
	publisher    = {Springer},
	address      = {Berlin, Heidelberg},
	series       = {Lecture {{Notes}} in {{Computer Science}}},
	pages        = {429--528},
	doi          = {10.1007/3-540-65306-6_21},
	isbn         = {978-3-540-49442-3},
	editor       = {Reisig, Wolfgang and Rozenberg, Grzegorz}
}

@inproceedings{kurshan98,
	title        = {Static Partial Order Reduction},
	author       = {Kurshan, R. and Levin, V. and Minea, M. and Peled, D. and Yenig{\" u}n, H.},
	year         = 1998,
	booktitle    = {Tools and {{Algorithms}} for the {{Construction}} and {{ Analysis}} of {{Systems}}},
	publisher    = {Springer},
	address      = {Berlin, Heidelberg},
	pages        = {345--357},
	doi          = {10.1007/BFb0054182},
	isbn         = {978-3-540-69753-4},
	editor       = {Steffen, Bernhard}
}

@inproceedings{kahlon09,
	title        = {Monotonic {{Partial Order Reduction}}: {{An Optimal Symbolic Partial Order Reduction Technique}}},
	author       = {Kahlon, Vineet and Wang, Chao and Gupta, Aarti},
	year         = 2009,
	booktitle    = {Computer {{Aided Verification}}},
	publisher    = {Springer},
	address      = {Berlin, Heidelberg},
	pages        = {398--413},
	doi          = {10.1007/978-3-642-02658-4_31},
	isbn         = {978-3-642-02658-4},
	editor       = {Bouajjani, Ahmed and Maler, Oded}
}

@book{rosenberg96,
	title        = {How Debuggers Work: Algorithms, Data Structures, and Architecture},
	author       = {Rosenberg, Jonathan B.},
	year         = 1996,
	month        = oct,
	publisher    = {John Wiley \& Sons, Inc.},
	address      = {USA},
	isbn         = {978-0-471-14966-8}
}

@inproceedings{potsch17,
	title        = {Advanced Remote Debugging of {{LoRa-enabled IoT}} Sensor Nodes},
	author       = {P{\"o}tsch, Albert and Haslhofer, Florian and Springer, Andreas},
	year         = 2017,
	month        = oct,
	booktitle    = {Proceedings of the {{Seventh International Conference}} on the {{Internet}} of {{Things}}},
	publisher    = {Association for Computing Machinery},
	address      = {New York, NY, USA},
	series       = {{{IoT}} '17},
	pages        = {1--2},
	doi          = {10.1145/3131542.3140259},
	isbn         = {978-1-4503-5318-2}
}

@inproceedings{skvar-c24,
	title        = {In-{{Field Debugging}} of {{Automotive Microcontrollers}} for {{ Highest System Availability}}},
	author       = {Skvar{\v c} Bo{\v z}i{\v c}, Ga{\v s}per and Irigoyen Ceberio, Ibai and Mayer, Albrecht},
	year         = 2024,
	month        = sep,
	booktitle    = {Proceedings of the 2nd {{ACM International Workshop}} on {{ Future Debugging Techniques}}},
	publisher    = {Association for Computing Machinery},
	address      = {New York, NY, USA},
	series       = {{{DEBT}} 2024},
	pages        = {2--8},
	doi          = {10.1145/3678720.3685314},
	isbn         = 9798400711107
}

@misc{soderby24,
	title        = {Debugging with the Arduino {{IDE}} 2.0},
	author       = {S{\"o}derby, Karl and De Feo, Ubi},
	year         = 2024,
	month        = nov,
	howpublished = {https://docs.arduino.cc/software/ide-v2/tutorials/ide-v2-debugger}
}

@article{shortm,
	title        = {{{IEEE Standard}} for {{Test Access Port}} and {{Boundary-Scan Architecture}}},
	author       = {},
	year         = 2013,
	month        = may,
	journal      = {IEEE Std 1149.1-2013 (Revision of IEEE Std 1149.1-2001)},
	pages        = {1--444},
	doi          = {10.1109/IEEESTD.2013.6515989}
}

@article{hogl06,
	title        = {Open On-Chip Debugger--Openocd--},
	author       = {H{\"o}gl, Hubert and Rath, Dominic},
	year         = 2006,
	journal      = {Fakultat fur Informatik, Tech. Rep},
	publisher    = {Citeseer}
}

@article{marra18,
	title        = {Out-{{Of-Place}} Debugging: A Debugging Architecture to Reduce Debugging Interference},
	author       = {Marra, Matteo and Polito, Guillermo and Gonzalez Boix, Elisa},
	year         = 2018,
	month        = nov,
	journal      = {The Art, Science, and Engineering of Programming},
	volume       = 3,
	number       = 2,
	pages        = {3:1--3:29},
	doi          = {10.22152/programming-journal.org/2019/3/3},
	issn         = {2473-7321}
}

@inproceedings{lauwaerts22,
	title        = {Event-{{Based Out-of-Place Debugging}}},
	author       = {Lauwaerts, Tom and Castillo, Carlos Rojas and Singh, Robbert Gurdeep and Marra, Matteo and Scholliers, Christophe and Gonzalez Boix, Elisa},
	year         = 2022,
	month        = nov,
	booktitle    = {Proceedings of the 19th {{International Conference}} on {{ Managed Programming Languages}} and {{Runtimes}}},
	publisher    = {Association for Computing Machinery},
	address      = {New York, NY, USA},
	series       = {{{MPLR}} '22},
	pages        = {85--97},
	doi          = {10.1145/3546918.3546920},
	isbn         = {978-1-4503-9696-7}
}

@article{bernstein95a,
	title        = {Operational {{Semantics}} of a {{Focusing Debugger}}},
	author       = {Bernstein, Karen L. and Stark, Eugene W.},
	year         = 1995,
	month        = jan,
	journal      = {Electronic Notes in Theoretical Computer Science},
	series       = {{{MFPS XI}}, {{Mathematical Foundations}} of {{Programming Semantics}}, {{Eleventh Annual Conference}}},
	volume       = 1,
	pages        = {13--31},
	doi          = {10.1016/S1571-0661(04)80002-1},
	issn         = {1571-0661}
}

@inproceedings{ferrari01,
	title        = {A Debugging Calculus for Mobile Ambients},
	author       = {Ferrari, GianLuigi and Tuosto, Emilio},
	year         = 2001,
	month        = mar,
	booktitle    = {Proceedings of the 2001 {{ACM}} Symposium on {{Applied}} Computing},
	publisher    = {ACM},
	address      = {Las Vegas Nevada USA},
	doi          = {10.1145/372202.380701},
	isbn         = {978-1-58113-287-8}
}

@inproceedings{torres17,
	title        = {A Principled Approach towards Debugging Communicating Event-Loops},
	author       = {Torres Lopez, Carmen and Boix, Elisa Gonzalez and Scholliers, Christophe and Marr, Stefan and M{\"o}ssenb{\"o}ck, Hanspeter},
	year         = 2017,
	month        = oct,
	booktitle    = {Proceedings of the 7th {{ACM SIGPLAN International Workshop}} on {{Programming Based}} on {{Actors}}, {{Agents}}, and {{ Decentralized Control}}},
	publisher    = {Association for Computing Machinery},
	address      = {New York, NY, USA},
	series       = {{{AGERE}} 2017},
	pages        = {41--49},
	doi          = {10.1145/3141834.3141839},
	isbn         = {978-1-4503-5516-2}
}

@inproceedings{holter24,
	title        = {Abstract {{Debuggers}}: {{Exploring Program Behaviors}} Using {{ Static Analysis Results}}},
	author       = {Holter, Karoliine and Hennoste, Juhan Oskar and Lam, Patrick and Saan, Simmo and Vojdani, Vesal},
	year         = 2024,
	month        = oct,
	booktitle    = {Proceedings of the 2024 {{ACM SIGPLAN International Symposium}} on {{New Ideas}}, {{New Paradigms}}, and {{Reflections}} on {{ Programming}} and {{Software}}},
	publisher    = {Association for Computing Machinery},
	address      = {New York, NY, USA},
	series       = {Onward! '24},
	pages        = {130--146},
	doi          = {10.1145/3689492.3690053},
	isbn         = 9798400712159
}

@article{baldoniSurveySymbolicExecution2019,
  title = {A {{Survey}} of {{Symbolic Execution Techniques}}},
  author = {Baldoni, Roberto and Coppa, Emilio and D’elia, Daniele Cono and Demetrescu, Camil and Finocchi, Irene},
  date = {2019-05-31},
  journaltitle = {ACM Computing Surveys},
  shortjournal = {ACM Comput. Surv.},
  volume = {51},
  number = {3},
  pages = {1--39},
  issn = {0360-0300, 1557-7341},
  doi = {10.1145/3182657},
  url = {https://dl.acm.org/doi/10.1145/3182657},
  urldate = {2025-01-24},
  langid = {english}
}

@inproceedings{liEmpiricalStudyConcurrency2023a,
  title = {An {{Empirical Study}} on {{Concurrency Bugs}} in {{Interrupt-Driven Embedded Software}}},
  booktitle = {Proceedings of the 32nd {{ACM SIGSOFT International Symposium}} on {{Software Testing}} and {{Analysis}}},
  author = {Li, Chao and Chen, Rui and Wang, Boxiang and Wang, Zhixuan and Yu, Tingting and Jiang, Yunsong and Gu, Bin and Yang, Mengfei},
  date = {2023-07-12},
  pages = {1345--1356},
  publisher = {ACM},
  location = {Seattle WA USA},
  doi = {10.1145/3597926.3598140},
  url = {https://dl.acm.org/doi/10.1145/3597926.3598140},
  urldate = {2025-04-25},
  eventtitle = {{{ISSTA}} '23: 32nd {{ACM SIGSOFT International Symposium}} on {{Software Testing}} and {{Analysis}}},
  isbn = {979-8-4007-0221-1},
  langid = {english}
}

@article{mio,
author = {Lauwaerts, Tom and Steevens, Maarten and Scholliers, Christophe},
title = {{MIO: Multiverse Debugging in the Face of Input/Output}},
year = {2025},
issue_date = {October 2025},
publisher = {Association for Computing Machinery},
address = {New York, NY, USA},
volume = {9},
number = {OOPSLA2},
url = {https://doi.org/10.1145/3763136},
doi = {10.1145/3763136},
journal = {Proc. ACM Program. Lang.},
month = oct,
articleno = {358},
numpages = {30}
}

\newpage
\clearpage

\appendix
\crefalias{section}{appendix}
\crefalias{subsection}{appendix}

\section{Additional server rules}
\label{appendix:additional-rules}
In this appendix we list a series of additional rules for the debug server.

\subsection{Running rules}
\label{app:running-rules}
In this section we list the rules used by the server during normal execution in the $\textsc{Running}$ mode. These can be found in \cref{fig:runServerRules}.
\begin{figure*}[!h]
	\begin{mathpar}
		\inferrule[(\textsc{Run})]
       		{
				\textsf{non-prim} \; K \\
				K = \{ \rho,\updelta,st,\mu,e:e^* \} \\
				id(e) \notin bps \\
				K \hookrightarrow_{i} K'\\
            }
            {
				\langle \textsc{Running}, \varnothing, \varnothing, bps, c_{instr}, K \rangle \hookrightarrow_{d,i} \langle \textsc{Running}, \varnothing, \varnothing, bps, c_{instr} + 1, K' \rangle
			}\\
	\inferrule[(\textsc{Run-Prim-Out})]
       		{
				K = \{ \rho,\updelta,v_a^*:v^*,\mu,e:e^* \} \\
				e = \textbf{call} \; j \\
				P^{Out}(j) = p \\
				\lfloor p(v_a^*) \rfloor \\
				id(e) \notin bps \\
            }
            {
				\langle \textsc{Running}, \varnothing, \varnothing, bps, c_{instr}, K \rangle
				\hookrightarrow_{d,i} 
				\langle \textsc{Running}, \varnothing, \varnothing, bps, c_{instr} + 1, \{ \rho,\updelta,v^*,\mu, e^* \} \rangle			
			}\\
		\inferrule[(\textsc{Run-Prim-In})]
       		{
				K = \{ \rho,\updelta,v_a^*:v^*,\mu,e:e^* \} \\
				e = \textbf{call} \; j \\
				P^{In}(j) = p \\
				K' = \{ \rho,\updelta,v:v^*,\mu, e^* \} \\
				v \in \lfloor p(v_a^*) \rfloor \\
				id(e) \notin bps \\
            }
            {
				\langle \textsc{Running}, \varnothing, \varnothing, bps, c_{instr}, K \rangle
				\hookrightarrow_{d,i} 
				\langle \textsc{Running}, \varnothing, \mathit{prim}(c_{instr}+1, v), bps, 0, K' \rangle			
			}\\
	\end{mathpar}
	\caption{Server rules used during execution without any breakpoints or incoming debug messages.}
	\label{fig:runServerRules}
\end{figure*}

The \textsc{Run} rule is used during normal operation when there are no primitives being used and there are no breakpoints.
At every step the VM will execute one instruction with the $K \hookrightarrow_{i} K'$ relation. 
In order for the server to efficiently communicate the state of the underlying VM to the client, we keep track of the number of instructions executed since the last synchronisation with the client,  therefore in the \textsc{Run} rule the $c_{instr}$ counter is incremented for each step taken.
We deliberately omit the details of how breakpoints are exactly represented in this semantics, i.e. we assume that there is a function $id$ which given an expression uniquely identifies that expression in the program.

When there is no breakpoint but a primitive is being called, two rules apply.  When the primitive is an output primitive, the 
\textsc{Run-Prim-Out} rules is used. When executing this primitive, the primitive will be executed outside of WebAssembly. Because these primitives perform an action and do not read any values from the environment we consider them deterministic.
If the primitive is an input primitive, the \textsc{Run-Prim-In} rule is used.
When executing an input primitive, the primitive will be executed outside of WebAssembly and the result will be returned to the VM. Because these primitives are executed outside of the VM and return a value for use inside the VM, they are considered the source of non-determinism. 

The server communicates to the client which value the nondeterministic operation returned with the $prim(c_{instr},v)$ message which captures the returned value $v$ and the number of instructions $c_{instr}$ executed since the last synchronisation with the server.

\subsection{Rules for handling breakpoints}
In this section we list a series of rules for adding and removing breakpoints. These can be found in \cref{fig:breakpoints}.

\begin{figure*}[h!]
	\begin{mathpar}
		\inferrule[(\textsc{Dbg-bp-Add})]
       		{
			}
            {
				\langle es, \mathit{break}^+ \; id, \varnothing, bps, c_{instr}, K \rangle
				\hookrightarrow_{d,i}
				\langle es, \varnothing, msg, bps \cup id, c_{instr}, K \rangle
			}\\
		\inferrule[(\textsc{Dbg-Bp-Rem})]
       		{
			}
            {
				\langle es, \mathit{break}^- \; id, \varnothing, bps, c_{instr}, K \rangle
				\hookrightarrow_{d,i}
				\langle es, \varnothing, msg, bps \backslash id, c_{instr}, K \rangle
			}\\
	\end{mathpar}
	\caption{Rules to add and remove breakpoints on the server.}
	\label{fig:breakpoints}
\end{figure*}

The rule \textsc{Dbg-Bp-Add} adds a breakpoint to the set of breakpoints $bps$ when the $break^+$ message is received. The rule \textsc{Dbg-Bp-Rem} removes a breakpoint from this set when the $break^-$ message is received. 

\subsection{Rules for pausing the server}
In this section we list two additional rules that are used to pause the server when it hits a breakpoint or to manually pause the server. These rules are listed in \cref{fig:pause-bp}.

The rule \textsc{Dbg-Break} says that if a breakpoint is present on the current instruction the server will go into the \textsc{Paused} state. 
Additionally the \textsc{Dbg-Pause} rule can be used to manually pause the execution. This rule says that when the $pause$ message is received from the client the server will switch from the \textsc{Running} to the \textsc{Paused} state.

\begin{figure*}[h!]
	\begin{mathpar}
		\inferrule[(\textsc{Dbg-Break})]
       		{
				K = \{ \rho,\updelta,st,\mu,e^* \} \\
				id(e^*) \in bp \\
            }
            {
				\langle \textsc{Running}, \varnothing, \varnothing, bps, c_{instr}, K \rangle
				\hookrightarrow_{d,i} 
				\langle \textsc{Paused}, \varnothing, \mathit{executed}(c_{instr}), bps, 0, K \rangle
			}\\
		\inferrule[(\textsc{Dbg-Pause})]
       		{
            }
            {
				\langle \textsc{Running}, \mathit{pause}, \varnothing, bps, c_{instr}, K \rangle
				\hookrightarrow_{d,i}
				\langle \textsc{Paused}, \varnothing, \mathit{executed}(c_{instr}), bps, 0, K \rangle
			}\\
	\end{mathpar}
	\caption{Semantic rules used to pause the debug server.}
	\label{fig:pause-bp}
\end{figure*}

\subsection{Additional rules for when paused}
The rule \textsc{Dbg-Play} shown in \cref{fig:additional-paused} is used to put the server back in the \textsc{Running} mode after a $play$ message is received. The \textsc{Dbg-Inspect} rule is used by the client to retrieve the current state $K$ from the server. This is useful when wanting to analyse the current state. Finally the \textsc{Dbg-Reset} rule is used to reset the execution of the VM. In this case the state $K$ is reset to $K_{start}$. The state $K_{start}$ is the start state in which the WebAssembly module is loaded.

\begin{figure*}[h!]
	\begin{mathpar}
		\inferrule[(\textsc{Dbg-Play})]
       		{ 
            }
            {
				\langle \textsc{Paused}, \mathit{play}, \varnothing, bp, c_{instr}, K \rangle \hookrightarrow_{d,i} \langle \textsc{Running}, \varnothing, \varnothing, bp, c_{instr}, K \rangle
			}\\
		\inferrule[(\textsc{Dbg-Inspect})]
       		{
				msg = snapshot(K)
            }
            {
				\langle \textsc{Paused}, \mathit{inspect}, \varnothing, bp, c_{instr}, K \rangle \hookrightarrow_{d,i} \langle \textsc{Paused}, \varnothing, msg, bp, c_{instr}, K \rangle
			}\\
		\inferrule[(\textsc{Dbg-Reset})]
       		{
            }
            {
				\langle \textsc{Paused}, \mathit{reset}, \varnothing, bp, c_{instr}, K \rangle \hookrightarrow_{d,i} \langle \textsc{Paused}, \varnothing, \varnothing, bp, 0, K_{start} \rangle
			}\\
	\end{mathpar}
	\caption{Additional server rules used when the debugger is paused.}
	\label{fig:additional-paused}
\end{figure*}

\subsection{Traverse rules}
\label{appendix:traversal-rules}
The $\Rightarrow_{Traverse}$ rules shown in \cref{fig:traverse-rules} are used by the client to follow and/or create new paths if they do not yet exist, they are shown in \cref{fig:traverse-rules}. The rules processes the input one edge at a time. If the edge exists, the \textsc{Traverse-Node-Existing} rules applies, otherwise the \textsc{Traverse-Node-New} rule is used to create a new edge.

\begin{figure}[!h]
	\begin{mathpar}
		\inferrule[(\textsc{Traverse-Node-Existing})]
       		{
				T_{curr} = \mathit{Node} \; c^* \\
				\exists (msg, T_{next}) \in c^*
            }
            {
				\langle msg : msg^*, T,T_{curr} \rangle
				\Rightarrow_{Traverse}
				\langle msg^*, T,T_{next} \rangle
			}\\
			\inferrule[(\textsc{Traverse-Node-New})]
       		{
				T_{curr} = \mathit{Node} \; c^* \\
				\nexists (msg, T_{child}) \in c^* \\
				T' = \mathit{attach}(T, T_{curr}, (msg, T_{next})) \\
				T_{next} =  \mathit{Node}~\varnothing \\
            }
            {
				\langle msg : msg^*, T,T_{curr} \rangle
				\Rightarrow_{Traverse}
				\langle msg^*, T',T_{next} \rangle
			}\\
	\end{mathpar}
	\caption{Rules needed to navigate and extend the tree.}
	\label{fig:traverse-rules}
\end{figure}

\subsection{Path rules}
\label{appendix:path-rules}
The $\Rightarrow_{Path}$ rules shown in \cref{fig:path-rules} are used by the client to determine the sequence of messages needed to traverse from one state to another. The path rule is defined recursively, with \textsc{Path-Equal} being the stop condition. The rule builds up the messages needed to traverse from one state to another step by step. The \textsc{Path-Not-Equal} rule finds the child $T_c$ of the node $T$ for which $T'$ the target node is a descendant. It then adds the edge label to the list of messages and then recursively adds the messages needed to complete the path starting from this child node $T_c$ to the target node.

\begin{figure}
	\begin{mathpar}
		\inferrule[(\textsc{Path-Equal})]
       		{
				id(T) = id(T')
            }
            {
				T, msg^* \Rightarrow_{Path} T', msg^*
			}\\
		\inferrule[(\textsc{Path-Not-Equal})]
       		{
				id(T) \neq id(T') \\
				T = \mathit{Node} \; c^* \\
				\exists (msg_c, T_c) \in c^*, T_c \vdash^* T' \\
				T_{c}, \varnothing \Rightarrow_{Path}^* T', msg^*_{tail}
            }
            {
				T, msg^* \Rightarrow_{Path} T', msg^{*} : msg_c : msg^*_{tail}
			}\\

	\end{mathpar}
	\caption{Semantics to determine the path to a destination}
	\label{fig:path-rules}
\end{figure}

\subsection{Concolic expansion rules}
\label{app:expansion-rules}
To make online concolic execution start from an existing state that is not concolic but concrete a transformation needs to be performed. To achieve this transformation, we introduce rules in \cref{fig:concolic-expansion-rules} for expanding the program's state into its symbolic counterpart. \textsc{S-Val} is used to take a value and convert it into a symbolic value. \textsc{S-List} converts lists such as the stack, memory, globals, and locals into symbolic representations. Finally, \textsc{S-Concrete} combines these rules to create symbolic representations for the locals, globals, stack, and memory. The combination of concrete and symbolic locals, globals, stack, memory, an empty symbolic environment, and a path condition initialized to $true$ forms a concolic state. On this concolic state, the concolic semantics ($\Rightarrow_{cs}$), described in \cref{sec:concolic-rules} can be applied allowing the debugger to generate a a tree containing the future execution paths of the program.

\begin{figure}[!h]
	\begin{mathpar}
        \inferrule[(\textsc{S-List})]
    	            { 
                x = v : v^* \\
                x' = s(v) : s(v^*) \\
            }
            {
            	s(x) = x'
            }
    
        \inferrule[(\textsc{S-Val})]
    	            {
	            }
            {
    						s(v) = v
            }
            
         \inferrule[(\textsc{S-Concrete})]
    	            { 
    						\hat{\rho} = s(\rho) \\
    						\hat{\updelta} = s(\updelta) \\
    						\hat{st} = s(st) \\
    						\hat{\mu} = s(\mu) \\
    						\varepsilon = \varnothing \\
    						\pi = true \\
							K = \{ \rho,\updelta,st,\mu,e^* \} \\
            }
            {
    						s(K) = \{ K,\hat{\rho},\hat{\updelta},\hat{st},\hat{\mu},\varepsilon,\pi \}
            }
	\end{mathpar}
	\caption{Rules for extending a concrete state into a concolic state.}
	\label{fig:concolic-expansion-rules}
\end{figure}

\section{Example Concolic execution}
\label{appendix:example-concolic-execution}
To illustrate the workings of concolic execution, we provide a small example. 
Consider the program shown in \cref{lst:example-concolic}, which reads one sensor value using \texttt{chip\_analog\_read}. 
Depending on this sensor value, an LED is turned on. Reading this sensor value generates a symbolic variable we call $x$. 
The first time the program is executed, the value of $x$ is chosen at random, for example $0$. 
The program is then executed with this value both concretely and symbolically. 
With $x = 0$, the condition $x < 5$ evaluates to $\texttt{true}$. 
Since the condition holds, the path condition to reach that state is $x < 5$. 
This remains the final path condition for this program because no other paths are explored in this initial execution.

Concolic execution aims to find all possible paths. 
To achieve this, the analysis employs an SMT solver to identify a value for $x$ that would result in a new execution path. 
Initially, only one path has been explored: $x < 5$. To discover additional paths, the algorithm instructs the SMT solver to solve the equation $\lnot(x < 5)$. 
A possible solution is $x = 5$.  
The analysis then re-executes the program with this new model. 
In this case, the condition evaluates to $\texttt{false}$, and the path condition becomes $\lnot(x < 5)$. 
After this execution, the analysis attempts to find yet another path by seeking values of $x$ that are neither less than $5$ nor greater than or equal to $5$. 
However, since no such value exists, the analysis concludes that all possible paths have been explored.
The resulting models $\{x = 0\}$ and $\{x = 5\}$ suffice to cover all branches of this program.

\begin{figure}[h!]
\begin{lstlisting}[language=C, style=CStyle,escapechar=']
void main() {
    if (chip_analog_read(SENSOR) < 5) {
        chip_digital_write(LED, HIGH);
    }
}
\end{lstlisting}
\caption{Example program to illustrate the operation of concolic execution.}
\label{lst:example-concolic}
\end{figure}

\section{Prototype implementation}\label{app:prototype}

We have implemented a prototype frontend for our remote concolic multiverse debugger in Kotlin.
\Cref{fig:example01} shows a screenshot of an active debugging session in the debugger.

\begin{figure*}[!h]
	\centering
	\includegraphics[width=\textwidth]{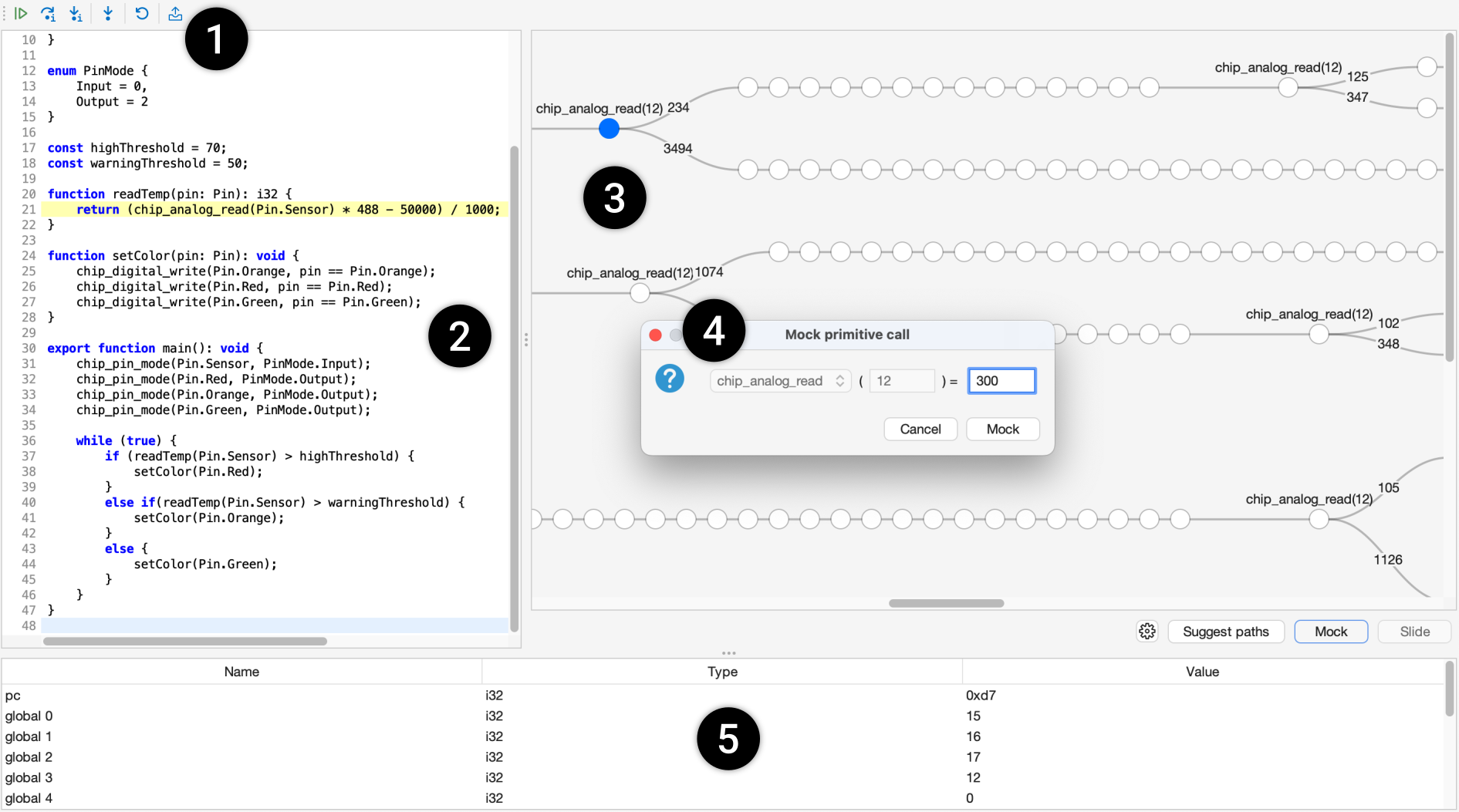}
	\caption{
	Screenshot of the proposed debugger debugging a small temperature sensor program.
        Top left (1): the debug operations; pause or continue, step over (instruction), step into (instruction), step to the next line, and update program.
        Left pane (2): source code.
        Top-right pane (3): the multiverse tree. 
        Popup window (4): window to mock a primitive to return an arbitrary value.
        Bottom pane (5): general debugging information, such as the global and local variables, and the current program counter.}
        \label{fig:example01}
\end{figure*}

\subsection{The available Debug Operations}
Our example debug scenario illustrates all the debug operations and features of our remote concolic multiverse debugger.
For clarity we summarize the operations here.

\subparagraph*{Suggest paths} At any point in the program developers can ask the debugger to perform concolic analysis and suggest future paths starting from the current node in the multiverse tree. The depth of these paths, can be configured by the user as well. The paths are added automatically to the tree visualisation.

\subparagraph*{Mock} Through the mock button, developers can register mock values for primitive calls. Whenever a call occurs with the right arguments, the value in the mock table will be used rather than any \emph{realtime} input value.

\subparagraph*{Step operations} The debugger has three different step instructions: the typical step over and single step (or step into) operations, and a \emph{step to next line} operation. In our prototype, the steps operate at the level of WebAssembly instructions, but the third operation allows for stepping to the next line of AssemblyScript code.
The step instructions will update the current line in de source code view, and the current node in the tree visualisation; adding any new nodes where necessary.

\subparagraph*{Slide} Aside from stepping forwards linearily with the step instructions, developers can select any node in the multiverse tree, including those generated by the concolic analysis, and slide to this node.
When sliding to a path that branched previous to the current node, the debugger will reset the program and deterministically replay the execution using the input mocking function.

\subparagraph*{Restart} Additionally, it is also possible to manually reset the execution to its  starting point without using slide with the restart button. This will keep the existing multiverse tree intact but jump to the root node, allowing other branches to be explored.

\subparagraph*{Pause \& play} Similar to conventional debugger, the debugger can pause and start execution at any time.

\subparagraph*{Upload program} Upload the program currently opened in the source code view, to the debugger backend.

\vfill
\newpage
\section{Debugger Guarantees}\label{app:guarantee} 

Debugger soundness and completeness is defined in terms of the relationship between the debugging semantics and the underlying language semantics.
We denote the start of the WebAssembly program as $K_{start}$ and the start of the debugging session as $dbg_{start}$.
The start states for the client and server are defined as follows.
$$
dbg_{start} = C_{start} \; | \; S_{start} = \langle \varnothing, \varnothing, \mathit{Node} \; \varnothing,\mathit{Node} \; \varnothing \rangle \; | \; \langle \textsc{Paused}, \varnothing, \varnothing, \varnothing, 0, K_{start}\rangle
$$

We now provide the full proofs for debugger soundness and completeness of the semantics presented in the main article.

\begin{theorem}[Debugger soundness]\label{theorem:debugger-soundness}
For every debugger state $dbg'$ that is reachable by performing debugging operations $\rightarrow$ starting from the start state of the client $dbg_{start}$, the program state $K'$ contained in $dbg'$ must be reachable through the underlying language semantics $\hookrightarrow_i$ by stepping from the start program state $K_{start}$.
$$
    \forall \; dbg', where \; K' \in  dbg', \;
    dbg_{start} \rightarrow^* dbg' \Rightarrow K_{start} \hookrightarrow_i^* K'
$$
\end{theorem}

\begin{proof}
This theorem holds through induction over the debug steps $\rightarrow^*$.

\noindent
\textbf{Base case} In the base case there are zero steps, and the theorem holds trivially because $K_{start}$ equals $K'$.

\noindent
\textbf{Induction step} In the induction step we know that $dbg_{start} \rightarrow^* dbg \Rightarrow K_{start} \hookrightarrow_i^* K \land K \in dbg$, and we now want to prove that the theorem still holds if we take the next step $dbg \rightarrow dbg' \Rightarrow K \hookrightarrow_i^* K' \land K' \in dbg'$. The next $dbg \rightarrow dbg'$ step can only be one of three cases shown in \cref{fig:serverClientCommunication}.
\begin{description}
\item[Case \textsc{Server-To-Client}] In case of the \textsc{Server-To-Client} rule, the client receives a message from the server and processes it. In this case all the rules leave the program state $K$ unchanged and so the theorem holds by the induction hypothesis.

\item[Case \textsc{Client-To-Server}] In the case of the \textsc{Client-To-Server} rule the server takes exactly one step from $S$ to $S'$, where the program state $K$ can change. We will discuss each possible server-side rule that can move $S$ to $S'$ in turn. The rules are shown in figures \ref{fig:server:paused}, \ref{fig:breakpoints}, \ref{fig:pause-bp}, and \ref{fig:additional-paused}.
\begin{description}
  \item[Subcase \textsc{Dbg-Step}] The new state $K'$ is the result of a single step $K \hookrightarrow_i K'$, and thereby the theorem holds since $K_{start} \hookrightarrow_i^* K$ holds by the induction hypothesis, combined with $K \hookrightarrow_i K'$, $K_{start} \hookrightarrow_i^* K'$ also holds.
 
  \item[Subcase \textsc{Dbg-Step-Prim-Out}, \textsc{Dbg-Step-Prim-In}] The debugging rules have the exact same preconditions as the \textsc{output-prim} and \textsc{input-prim} rules\footnote{These rules are shown in \cref{fig:primitive-execution}.} respectively.
  Given the same preconditions the debugging rules move $K$ to $K'$ in exactly the same way as their respective rules in the WebAssembly rules. Hence this case holds.

\item[Subcase \textsc{Dbg-Mock}] This rule only applies when $K$ can take a step with the rule \textsc{input-prim}. Instead of applying this rule, mocking replaces the call to primitive $p$ by a user-supplied value $v$.
  However, this value must be part of the codomain of the primitive. Since $mock$ is restricted to only mocking values that could actually be returned by an input primitive in the real execution, all states explorable in the debugger also exist in the underlying language semantics.
  In other words, the resulting state $\{ \rho,\updelta, v : v^*,\mu, e^* \}$ can also be reached by the nondeterministic rule \textsc{input-prim}, and hence the theorem holds in this case.

\item[Subcase \textsc{Dbg-bp-Add}, \textsc{Dbg-bp-Rem}, \textsc{Dbg-Pause}, \textsc{Dbg-Play}, \textsc{Dbg-Inspect}] The state $K$ is not changed and so the theorem holds by the induction hypothesis.

\item[Subcase \textsc{Dbg-Reset}] The execution restarts from the beginning, changing $K$ to $K_{start}$. Since $K_{start}$ can go in zero steps to itself, the theorem holds.

\end{description}

\item[Case \textsc{Server-Step}] In the last case the \textsc{Server-Step} rule, only the rules \textsc{Run}, \textsc{Run-Prim-Out}, \textsc{Run-Prim-In}, and \textsc{Dbg-Break} are possible since all message queues need to be empty. These rules are shown in  \ref{fig:runServerRules} and \ref{fig:pause-bp}.
\begin{description}
	\item[Case \textsc{Run}] For this rule, the new state K' is the result of a single step $K \hookrightarrow_i K'$. This is equivalent to the \textsc{Dbg-Step} case for \textsc{Server-To-Client}.
	\item[Case \textsc{Run-Prim-Out} and \textsc{Run-Prim-In}] These cases hold for the same reason as the \textsc{Dbg-Step-Prim-Out} and \textsc{Dbg-Step-Prim-In} rules for \textsc{Server-To-Client}.
\end{description}

\end{description}

\end{proof}

\begin{theorem}[Debugger completeness]\label{theorem:debugger-completeness} For every program state $K'$ reachable through the underlying language semantics $\hookrightarrow_i$ starting from start state $K_{start}$, there must be a corresponding debugging session where $dbg'$ is reachable starting from $dbg_{start}$ through $\rightarrow$ where the program state $K'$ is contained in $dbg'$.
$$
    \forall \; K', where \; K' \in  dbg', \;
    K_{start} \hookrightarrow_i^* K' \Rightarrow dbg_{start} \rightarrow^* dbg'
$$
\begin{proof}
This can be proven by induction over the steps $\hookrightarrow^*_{i}$.
  
  \noindent
  \textbf{Base case}
  In the base case there are zero steps, and the theorem holds trivially since $dbg_{start}$ contains $K_{start}$ by definition.
  
  \noindent
  \textbf{Induction step}
  In the induction step we know that that $K_{start} \hookrightarrow_i^* K \Rightarrow dbg_{start} \rightarrow^* dbg \land K \in dbg$. We now prove that the theorem holds for the next step $K \hookrightarrow_i K' \Rightarrow dbg \rightarrow^* dbg' \land K' \in dbg'$, there are only two cases to consider.
  \begin{description}
    \item[Case deterministic instructions] These instructions, which can be either normal WebAssembly instructions or calls to output primitives, are executed using $K \hookrightarrow_{i} K'$ in the underlying language semantics. The debugger semantics can use the $step$ message to move $K$ exactly to $K'$ with the \textsc{Dbg-Step} rule since this rule uses $\hookrightarrow_{i}$ from the underlying language semantics. Hence the theorem holds in this case.
    \item[Case nondeterministic instructions] For these instructions which are calls to input primitives, where $K = \{ \rho,\updelta,v_a^* : st,\mu,call \; j \}$ and $P^{In}(j) = p$ holds, \textsc{Input-Prim} applies in the underlying language semantics. We know which nondeterministic value $v$ was returned since we know $K'$. We can use the $\mathit{mock} \; v$ message in this case, so the \textsc{Dbg-Mock} step rule in the debugger semantics changes $K$ exactly into $K'$. As a result, the theorem also holds in this case.
  \end{description} 
\end{proof}
\end{theorem}

It should be noted that while the concolic execution used to expand the multiverse tree performs an under-approximation, it does not break completeness since any missed execution paths can still be reached through \textit{mock} messages.

\section{Example programs}

In this appendix we list the sources for each of the programs used in our evaluation. Additionally, we list the source code for two examples for which the code is not publicly available because it was written by ourselves or because it is a modification from an existing program.

\subsection{Sources}
\label{app:sources}
In \cref{tbl:program-urls} below, we list the sources for each of the programs we used in our evaluation. The source code for arduino-while-no-callibrate and the breakout game is provided below.

\begin{table}[h!]
\begin{tabular}{l | p{10cm}}
	Program & Source \\
	\hline
	arduino-knock & \url{https://github.com/arduino/arduino-examples/blob/main/examples/06.Sensors/Knock/Knock.ino} \\
	arduino-touch-sensor-lamp & \url{https://github.com/arduino/arduino-examples/blob/main/examples/10.StarterKit_BasicKit/p13_TouchSensorLamp/p13_TouchSensorLamp.ino} \\
    arduino-keyboard & \url{https://github.com/arduino/arduino-examples/blob/main/examples/10.StarterKit_BasicKit/p07_Keyboard/p07_Keyboard.ino} \\
    arduino-switch & \url{https://github.com/arduino/arduino-examples/blob/main/examples/05.Control/switchCase/switchCase.ino} \\
    arduino-love-o-meter & \url{https://github.com/arduino/arduino-examples/blob/main/examples/10.StarterKit_BasicKit/p03_LoveOMeter/p03_LoveOMeter.ino} \\
    arduino-crystal-ball &\url{https://github.com/arduino/arduino-examples/blob/main/examples/10.StarterKit_BasicKit/p11_CrystalBall/p11_CrystalBall.ino} \\
    arduino-knock-lock & \url{https://github.com/arduino/arduino-examples/blob/main/examples/10.StarterKit_BasicKit/p12_KnockLock/p12_KnockLock.ino} \\
    arduino-zoetrope & \url{https://github.com/arduino/arduino-examples/blob/main/examples/10.StarterKit_BasicKit/p10_Zoetrope/p10_Zoetrope.ino} \\
    arduino-while & \url{https://github.com/arduino/arduino-examples/blob/main/examples/05.Control/WhileStatementConditional/WhileStatementConditional.ino} \\
	gesture-robot & \url{https://www.electronicsforu.com/electronics-projects/build-clap-gesture-controlled-robot} \\
\end{tabular}
\caption{Table listing the url for each of the programs used}
\label{tbl:program-urls}
\end{table}

\subsection{arduino-while-no-calibrate}
\label{app:while-no-calibrate}
In this section we list the source code for the arduino-while-no-calibrate example used in \cref{sec:evaluation}. This is example is a modification from the original arduino-while (\href{https://github.com/arduino/arduino-examples/blob/main/examples/05.Control/WhileStatementConditional/WhileStatementConditional.ino}{Source}) example program with the calibration step removed. Just like the other arduino examples the source code was also translated into AssemblyScript for use on the WARDuino VM.
\begin{lstlisting}[language=AssemblyScript, style=CStyle,breaklines=true]
@external("env", "chip_delay") declare function chip_delay(value: i32): void;
@external("env", "chip_pin_mode") declare function chip_pin_mode(pin: i32, mode: i32): void;
@external("env", "chip_digital_write") declare function chip_digital_write(pin: i32, value: i32): void;
@external("env", "chip_digital_read") declare function chip_digital_read(pin: i32): i32;
@external("env", "chip_analog_read") declare function chip_analog_read(pin: i32): i32;
@external("env", "chip_analog_write") declare function chip_analog_write(pin: i32, value: i32): void;

const INPUT = 0;
const OUTPUT = 2;

enum Pin {
    sensorPin = 0,
    ledPin = 9,
    indicatorLedPin = 13,
    buttonPin = 2,
}

let sensorMin: i32 = 1023;  // minimum sensor value
let sensorMax: i32 = 2048;     // maximum sensor value
let sensorValue: i32 = 0;   // the sensor value

export function main(): void {
    chip_pin_mode(Pin.indicatorLedPin, OUTPUT);
    chip_pin_mode(Pin.ledPin, OUTPUT);
    chip_pin_mode(Pin.buttonPin, INPUT);

    while (true) {
        // signal the end of the calibration period
        chip_digital_write(Pin.indicatorLedPin, 0);

        // read the sensor:
        sensorValue = chip_analog_read(Pin.sensorPin);

        // apply the calibration to the sensor reading
        sensorValue = (sensorValue - sensorMin) * 255 / (sensorMax - sensorMin);

        // in case the sensor value is outside the range seen during calibration
        if (sensorValue < 0) sensorValue = 0;
        if (sensorValue > 255) sensorValue = 255;

        // fade the LED using the calibrated value:
        chip_analog_write(Pin.ledPin, sensorValue);
    }
}
\end{lstlisting}

\subsection{Breakout game}
\label{app:breakout}
In this section we list the source code for the breakout game used in \cref{sec:evaluation} and shown in \cref{fig:breakout-game}.

\begin{figure}
	\centering
	\includegraphics[width=0.5\linewidth]{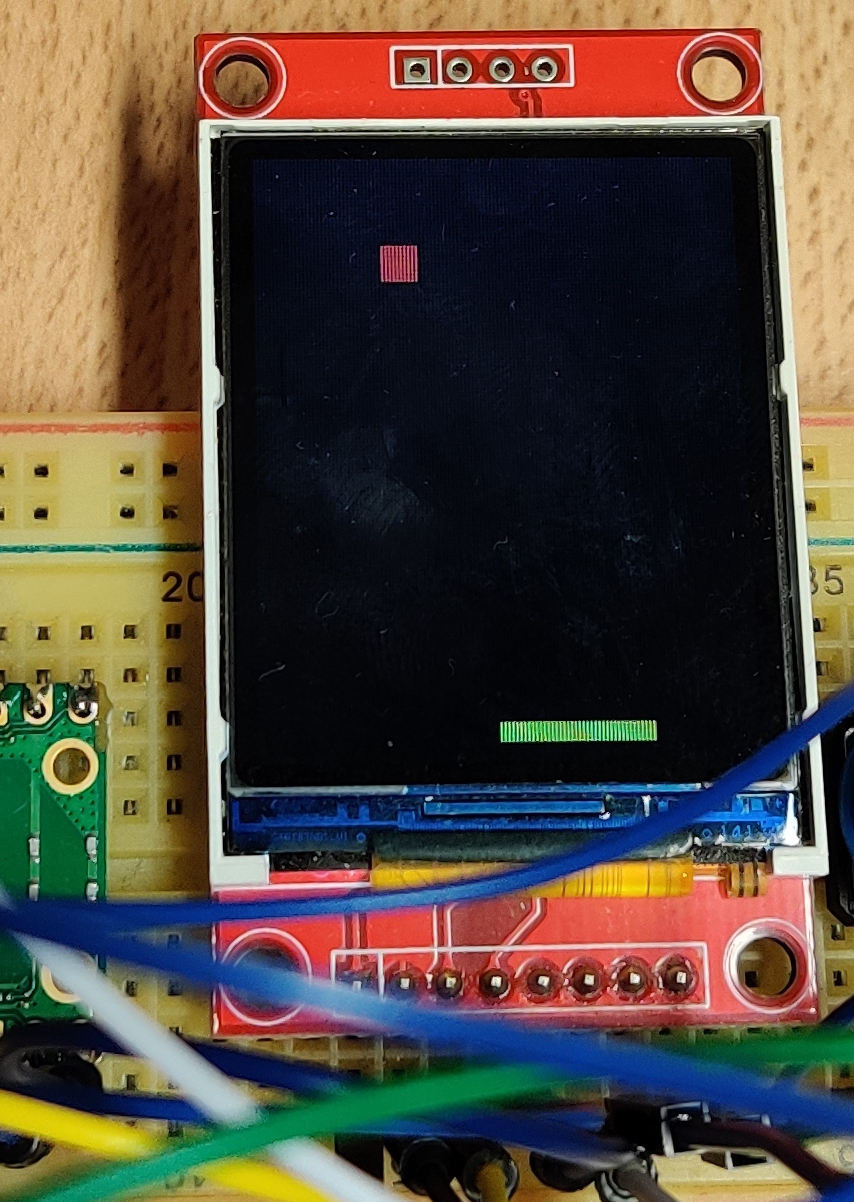}
	\caption{The breakout game written in AssemblyScript running on the modified WARDuino VM on a Raspberry Pi Pico microcontroller with an ST7735 display.}
	\label{fig:breakout-game}
\end{figure}

\begin{lstlisting}[language=AssemblyScript, style=CStyle,breaklines=true]
// ST7735 Constants
const ST7735_NOP = 0x0
const ST7735_SWRESET = 0x01
const ST7735_RDDID = 0x04
const ST7735_RDDST = 0x09
const ST7735_SLPIN = 0x10
const ST7735_SLPOUT = 0x11
const ST7735_PTLON = 0x12
const ST7735_NORON = 0x13
const ST7735_INVOFF = 0x20
const ST7735_INVON = 0x21
const ST7735_DISPOFF = 0x28
const ST7735_DISPON = 0x29
const ST7735_CASET = 0x2A
const ST7735_RASET = 0x2B
const ST7735_RAMWR = 0x2C
const ST7735_RAMRD = 0x2E
const ST7735_COLMOD = 0x3A
const ST7735_MADCTL = 0x36
const ST7735_FRMCTR1 = 0xB1
const ST7735_FRMCTR2 = 0xB2
const ST7735_FRMCTR3 = 0xB3
const ST7735_INVCTR = 0xB4
const ST7735_DISSET5 = 0xB6
const ST7735_PWCTR1 = 0xC0
const ST7735_PWCTR2 = 0xC1
const ST7735_PWCTR3 = 0xC2
const ST7735_PWCTR4 = 0xC3
const ST7735_PWCTR5 = 0xC4
const ST7735_VMCTR1 = 0xC5
const ST7735_RDID1 = 0xDA
const ST7735_RDID2 = 0xDB
const ST7735_RDID3 = 0xDC
const ST7735_RDID4 = 0xDD
const ST7735_PWCTR6 = 0xFC
const ST7735_GMCTRP1 = 0xE0
const ST7735_GMCTRN1 = 0xE1
// Constants
const JOYSTICK = 2
// PIN configurations
const CS = 13
const RESET = 14
const DC = 15
const SDA = 7
const SCK = 6
// Arduino constants
const OUTPUT = 1
const INPUT = 0
const LOW = 0
const HIGH = 1
// Type declarations
@external("env", "write_spi_byte") 
declare function write_spi_byte(byte: i32): void;

@external("env", "write_spi_bytes_16") 
declare function write_spi_bytes_16(times: i32, color: i32): void;

@external("env", "chip_digital_write") 
declare function chip_digital_write(pin: i32, value: i32): void;

@external("env", "chip_digital_read") 
declare function chip_digital_read(pin: i32): i32;

@external("env", "chip_analog_read") 
declare function chip_analog_read(pin: i32): i32;

@external("env", "chip_pin_mode") 
declare function chip_pin_mode(pin: i32, value: i32): void;

@external("env", "chip_delay_us") 
declare function chip_delay_us(value: i32): void;

// Writing LCD pins
function LCD_SCK(b: i32): void { chip_digital_write(SCK, b); }

function LCD_SDO(b: i32): void { chip_digital_write(SDA, b); }

function LCD_RS(b: i32): void { chip_digital_write(DC, b); }

function LCD_CS(b: i32): void { chip_digital_write(CS, b); }

function LCD_RESET(b: i32): void { chip_digital_write(RESET, b); }

function writecommand(c: i32): void {
    LCD_RS(LOW);
    LCD_CS(LOW);
    write_spi_byte(c);
    LCD_CS(HIGH);
}

function writedata(c: i32): void {
    LCD_RS(HIGH);
    LCD_CS(LOW);
    write_spi_byte(c);
    LCD_CS(HIGH);
}

function setAddrWindow(x0: i32, y0: i32, x1: i32, y1: i32): void {
    writecommand(ST7735_CASET);  // column addr set
    writedata(0x00);
    writedata(x0 + 0);  // XSTART
    writedata(0x00);
    writedata(x1 + 0);  // XEND

    writecommand(ST7735_RASET);  // row addr set
    writedata(0x00);
    writedata(y0 + 0);  // YSTART
    writedata(0x00);
    writedata(y1 + 0);  // YEND

    writecommand(ST7735_RAMWR);  // write to RAM
}

const SCREEN_WIDTH = 128
const SCREEN_HEIGHT = 160

function chip_fill_rect(x: i32, y: i32, w: i32, h:i32, color: i32): void {
    setAddrWindow(x, y, x + w - 1, y + h - 1);
    // setup for data
    LCD_RS(HIGH);
    LCD_CS(LOW);
    write_spi_bytes_16(w * h, color);
    LCD_CS(HIGH);
}

function chip_fill_screen(color: i32): void {
    setAddrWindow(0, 0, SCREEN_WIDTH - 1, SCREEN_HEIGHT - 1);
    // setup for data
    LCD_RS(HIGH);
    LCD_CS(LOW);
    write_spi_bytes_16(SCREEN_WIDTH * SCREEN_HEIGHT, color);
    /*unsigned char colorB = color >> 8;
    for (int x=0; x < SCREEN_WIDTH*SCREEN_HEIGHT ; x++) {
            write_spi_byte(colorB);
            write_spi_byte(color);
    }*/
    LCD_CS(HIGH);
}

function ST7735_initR(): void {
    LCD_RESET(HIGH);
    chip_delay_us(500);
    LCD_RESET(LOW);
    chip_delay_us(500);
    LCD_RESET(HIGH);
    chip_delay_us(500);
    LCD_CS(LOW);

    writecommand(ST7735_SWRESET);  // software reset
    chip_delay_us(150);

    writecommand(ST7735_SLPOUT);  // out of sleep mode
    chip_delay_us(500);

    writecommand(ST7735_COLMOD);  // set color mode
    writedata(0x05);              // 16-bit color
    chip_delay_us(10);

    writecommand(ST7735_FRMCTR1);  // frame rate control - normal mode
    writedata(0x01);  // frame rate = fosc / (1 x 2 + 40) * (LINE + 2C + 2D)
    writedata(0x2C);
    writedata(0x2D);

    writecommand(ST7735_FRMCTR2);  // frame rate control - idle mode
    writedata(0x01);  // frame rate = fosc / (1 x 2 + 40) * (LINE + 2C + 2D)
    writedata(0x2C);
    writedata(0x2D);

    writecommand(ST7735_FRMCTR3);  // frame rate control - partial mode
    writedata(0x01);               // dot inversion mode
    writedata(0x2C);
    writedata(0x2D);
    writedata(0x01);  // line inversion mode
    writedata(0x2C);
    writedata(0x2D);

    writecommand(ST7735_INVCTR);  // display inversion control
    writedata(0x07);              // no inversion

    writecommand(ST7735_PWCTR1);  // power control
    writedata(0xA2);
    writedata(0x02);  // -4.6V
    writedata(0x84);  // AUTO mode

    writecommand(ST7735_PWCTR2);  // power control
    writedata(0xC5);              // VGH25 = 2.4C VGSEL = -10 VGH = 3 * AVDD

    writecommand(ST7735_PWCTR3);  // power control
    writedata(0x0A);              // Opamp current small
    writedata(0x00);              // Boost frequency

    writecommand(ST7735_PWCTR4);  // power control
    writedata(0x8A);              // BCLK/2, Opamp current small & Medium low
    writedata(0x2A);

    writecommand(ST7735_PWCTR5);  // power control
    writedata(0x8A);
    writedata(0xEE);

    writecommand(ST7735_VMCTR1);  // power control
    writedata(0x0E);

    writecommand(ST7735_INVOFF);  // don't invert display

    writecommand(ST7735_MADCTL);  // memory access control (directions)

    // http://www.adafruit.com/forums/viewtopic.php?f=47&p=180341

    // R and B byte are swapped
    // madctl = 0xC8;

    // normal R G B order
    // madctl = 0xC0;
    writedata(0xC0);  // row address/col address, bottom to top refresh

    writecommand(ST7735_COLMOD);  // set color mode
    writedata(0x05);              // 16-bit color

    writecommand(ST7735_CASET);  // column addr set
    writedata(0x00);
    writedata(0x00);  // XSTART = 0
    writedata(0x00);
    writedata(0x7F);  // XEND = 127

    writecommand(ST7735_RASET);  // row addr set
    writedata(0x00);
    writedata(0x00);  // XSTART = 0
    writedata(0x00);
    writedata(0x9F);  // XEND = 159
    writecommand(ST7735_GMCTRP1);
    writedata(0x0f);
    writedata(0x1a);
    writedata(0x0f);
    writedata(0x18);
    writedata(0x2f);
    writedata(0x28);
    writedata(0x20);
    writedata(0x22);
    writedata(0x1f);
    writedata(0x1b);
    writedata(0x23);
    writedata(0x37);
    writedata(0x00);
    writedata(0x07);
    writedata(0x02);
    writedata(0x10);
    writecommand(ST7735_GMCTRN1);
    writedata(0x0f);
    writedata(0x1b);
    writedata(0x0f);
    writedata(0x17);
    writedata(0x33);
    writedata(0x2c);
    writedata(0x29);
    writedata(0x2e);
    writedata(0x30);
    writedata(0x30);
    writedata(0x39);
    writedata(0x3f);
    writedata(0x00);
    writedata(0x07);
    writedata(0x03);
    writedata(0x10);
    writecommand(ST7735_DISPON);
    chip_delay_us(100);
    writecommand(ST7735_NORON);  // normal display on
    chip_delay_us(10);
    LCD_CS(HIGH);
}

@unmanaged
class Coord {
    x: i32;
    y: i32;
}

function max(a: i32, b: i32): i32 { return a > b ? a : b; }

function min(a: i32, b: i32): i32 { return a < b ? a : b; }
function diff(a: i32, b: i32): i32 { return a < b ? b-a : a-b; }

const DARK = 0X3DF7
const RED = 0x7C00
const GREEN = 0x03F0
const BLUE = 0x000F
const WHITE = 0xFFFF
const BLACK = 0x0000


const PURPLE = (RED | BLUE)
const CYAN = (BLUE | GREEN)
const YELLOW = (GREEN | RED)

// Ball Abstraction.
const BSIZE = 10
const PSIZE = 40
const PYPOS = 150
const BCOL = RED
const BCOL2 = CYAN
const BGCOL = (GREEN & DARK)
const PCOL = YELLOW

@unmanaged
class Ball {
    pos: Coord;
    speed: Coord;
}

export function main(): void {
    // ball state
    let ball: Ball = {
        pos: {
            x: 1000,
            y: 1000,
        },
        speed: {
            x: 50,
            y: 50
        }
    };
    // end init ball
    // end ball state
    // pin mode
    chip_pin_mode(CS, OUTPUT);
    chip_pin_mode(DC, OUTPUT);
    chip_pin_mode(SDA, OUTPUT);
    chip_pin_mode(SCK, OUTPUT);
    chip_pin_mode(RESET, OUTPUT);
    chip_pin_mode(JOYSTICK, INPUT);

    ST7735_initR();
    chip_fill_screen(BGCOL);
    let x = 64 - 20; // padle pos


    while (1) {
        // update ball
        let nextX = ball.pos.x + ball.speed.x;
        if(ball.speed.x > 0){
            chip_fill_rect(
               ball.pos.x/ 100, ball.pos.y / 100,
                diff(ball.pos.x/100,nextX/100), BSIZE,
                BGCOL);
        } else {
            chip_fill_rect(
                nextX/100 + BSIZE, ball.pos.y / 100,
                diff(ball.pos.x/100,nextX/100), BSIZE,
                BGCOL);
        }
        let nextY = ball.pos.y + ball.speed.y;
        if(ball.speed.y > 0){
            chip_fill_rect(
                ball.pos.x/ 100, ball.pos.y / 100,
                BSIZE,diff(ball.pos.y/100,nextY/100),
                BGCOL);
        } else {
            chip_fill_rect(
                ball.pos.x / 100, nextY/100 + BSIZE, //
                BSIZE,diff(ball.pos.y/100,nextY/100),
                BGCOL); //

        }
        ball.pos.x += ball.speed.x;
        ball.pos.y += ball.speed.y;
        chip_fill_rect(ball.pos.x / 100, ball.pos.y / 100, BSIZE, BSIZE,
                       ball.pos.y/100 + BSIZE >= PYPOS ? BCOL2 : BCOL);
        // bounce of the wall
        if (ball.pos.x + BSIZE * 100 > 12800) {
            ball.speed.x *= -1;
        }
        if (ball.pos.y + BSIZE * 100 > 16400) {
            ball.speed.y *= -1;
        }
        if (ball.pos.y < 0) {
            ball.speed.y *= -1;
        }
        if (ball.pos.x < 0) {
            ball.speed.x *= -1;
        }
        // end update ball
        // update paddle

        if (ball.pos.y/100 + BSIZE >= PYPOS) {
            if (ball.pos.x/100 + BSIZE > x && ball.pos.x/100 < x + PSIZE && ball.speed.y > 0) {
                ball.speed.y *= -1;
                ball.speed.y += ball.speed.y < 0 ? -1 : 1;
                ball.speed.x += ball.speed.x < 0 ? -1 : 1;

            }
        }

        chip_fill_rect(x, PYPOS, PSIZE, 5, BGCOL);
        x = chip_analog_read(JOYSTICK);
        chip_fill_rect(x, PYPOS, PSIZE, 5, PCOL);

        chip_delay_us(50);
    }
}

\end{lstlisting}

\section{Performance comparison}
\label{app:performance}
In this appendix we compare the performance overhead of our trace-based approach with MIO's snapshot-based approach on the breakout game. This game, due to its frequent usage of output primitives to control the display is very slow under MIO~\cite{mio}. This is caused by MIO taking snapshots after every primitive call which in this program is around every 2 or 3 instructions at certain points. When taking snapshots that frequently, performance degrades heavily, as shown in the MIO article. However, by using our trace-based approach to multiverse debugging we have very little overhead when compared to the normal execution. We believe this provides an interesting alternative implementation strategy for concrete multiverse debuggers that trades reversibility for performance.

To test the overhead of the different approaches, we tested how long it takes to execute 50 000 instructions of the breakout game without any snapshotting or tracing, with MIO's snapshotting strategy and our trace-based approach. We ran these tests on a Raspberry Pi Pico W which has a dual core Arm Cortex-M0+ processor running at 133MHz.

In \cref{fig:performance-full} we show the normal execution alongside the trace and snapshot-based approaches. Here it is clear that MIO's checkpointing strategy has significantly more overhead in this particular program. On the other hand, the trace-based approach is much closer to the normal execution because way less data needs to be sent to the client on the connected computer. Since the difference between the normal execution and the trace-based one is not clear in this figure, we also provide \cref{fig:performance-zoomed} which zooms in on just these two. This figure shows that the trace-based approach does have some overhead, but the overhead is relatively limited compared to the normal execution. While the overhead of the execution with tracing is relatively limited in comparison to MIO's checkpointing approach, it is important to note that the step back performance of our approach is reduced since the program needs to be re-executed from the start.

\definecolor{graphGreen}{HTML}{32A852}

\begin{figure}
\centering
\begin{tikzpicture}
\begin{axis}[
    width=\linewidth,
    height=7cm,
    xlabel={Instructions executed},
    ylabel={Time (s)},
    legend pos=north west,
    grid=major,
    mark options={solid},
]

\addplot[semithick,graphGreen,mark=square*] coordinates {
    (10000, 0.606)
    (20000, 0.941)
    (30000, 1.279)
    (40000, 1.611)
    (50000, 1.939)
};
\addlegendentry{Normal execution}

\addplot[semithick,blue!50,mark=*] coordinates {
    (10000, 0.606)
    (20000, 0.962)
    (30000, 1.321)
    (40000, 1.663)
    (50000, 2.025)
};
\addlegendentry{Tracing}

\addplot[semithick,red!50,mark=*] coordinates {
    (10000, 60.153)
    (20000, 120.273)
    (30000, 180.413)
    (40000, 240.746)
    (50000, 301.053)
};
\addlegendentry{Checkpointing}

\end{axis}
\end{tikzpicture}
\caption{Comparison of the normal execution, our trace-based approach and MIO's checkpointing approach.}
\label{fig:performance-full}
\end{figure}

\begin{figure}
\centering
\begin{tikzpicture}
\begin{axis}[
    width=\linewidth,
    height=7cm,
    xlabel={Instructions},
    ylabel={Time (ms)},
    legend pos=north west,
    grid=major,
    mark options={solid},
]

\addplot[semithick,graphGreen,mark=square*] coordinates {
    (10000, 606)
    (20000, 941)
    (30000, 1279)
    (40000, 1611)
    (50000, 1939)
};
\addlegendentry{Normal execution}

\addplot[semithick,blue!50,mark=*] coordinates {
    (10000, 606)
    (20000, 962)
    (30000, 1321)
    (40000, 1663)
    (50000, 2025)
};
\addlegendentry{Tracing}

\end{axis}
\end{tikzpicture}
\caption{Comparison of only the normal execution and our trace-based approach.}
\label{fig:performance-zoomed}
\end{figure}

\end{document}